\documentclass[a4paper,USenglish,cleveref,autoref,thm-restate]{lipics-v2021}

\usepackage{xspace}
\usepackage{xcolor}
\usepackage[colorinlistoftodos]{todonotes}
\usepackage{mathtools}
\usepackage[vlined,ruled,linesnumbered]{algorithm2e}
\SetAlgoSkip{}
\SetNlSty{textbf}{\color{gray}}{} 

\usepackage{cite} \usepackage{fnpct} \usepackage{microtype} 

\newcommand{\pc}{color-constrained PC-tree\xspace}
\newcommand{\pcs}{color-constrained PC-trees\xspace}

\newcommand{\Pcs}{Color-Constrained PC-Trees\xspace}

\newcommand{\peplan}{{\normalfont\textsc{Par}\-\textsc{tial}\-\textsc{ly} \textsc{Em}\-\textsc{bed}\-\textsc{ded} \textsc{Pla}\-\textsc{nari}\-\textsc{ty}}\xspace}

\newcommand{\Emb}{\ensuremath{\Omega}}
\newcommand{\Ord}{\ensuremath{\omega}}

\NewDocumentCommand \merge    {mmO{\ell}} {\ensuremath{{#1} \otimes_{#3} {#2}}}
\NewDocumentCommand \update   {mm} {\ensuremath{{#1} + {#2}}}
\NewDocumentCommand \setto  {mmO{\ell}} {
  \ensuremath{{#1} [{#2} \to {#3}]}}
\NewDocumentCommand \usplit   {mmO{#2}O{\ell}} {
  \ensuremath{({#1} + {#2})[{#3} \to {#4}]}}

\newcommand{\Merge}{\texttt{Merge}\xspace}
\newcommand{\Update}{\texttt{Up\-date}\xspace}
\newcommand{\Split}{\texttt{Split}\xspace}
\newcommand{\Intersect}{\texttt{In\-ter\-sect}\xspace}
\newcommand{\Project}{\texttt{Pro\-ject}\xspace}

\newcommand{\TCj}[1][j]{\ensuremath{\mathcal T[C_{#1}]}}
\newcommand{\Sj}[1][j]{\ensuremath{S_{#1}}}

\newcommand\molap[2]{\mathrlap{\ensuremath{#1}}\phantom{\ensuremath{#2}}}
\newcommand\negphantom[1]{\settowidth{\dimen0}{#1}\hspace*{-\dimen0}}

\makeatletter
\NewDocumentCommand\namedlabel{omm}{\begingroup \def\@currentlabel{#3}\phantomsection \IfNoValueTF{#1}
  {\label{#2}}
  {\label[#1]{#2}}\endgroup }
\makeatother

\newtoggle{long}\toggletrue{long} \newtoggle{myapx}\togglefalse{myapx}
\NewDocumentEnvironment{prooflater}{m +b}{\expandafter\global\expandafter\def\csname#1\endcsname{\begin{proof}#2\end{proof}}}{\ignorespacesafterend}
\NewDocumentEnvironment{proofsketch}{O{Proof sketch.}}{\begin{proof}[#1]}{\end{proof}\ignorespacesafterend}
\newcommand{\AppendixSymbol}{\ensuremath{\star}}
\newcommand{\restateref}[1]{[\iftoggle{myapx}{\hyperref[#1]{\AppendixSymbol}}{\hyperref[#1*]{\AppendixSymbol}}]}
\NewDocumentEnvironment{statelater}{m +b}{\expandafter\global\expandafter\def\csname#1\endcsname{#2}}{\ignorespacesafterend}

\let\oldrestatable\restatable
\def\restatable{\expandafter\oldrestatable}
\makeatletter
\pretocmd{\thmt@rst@storecounters}{\Hy@SaveLastskip}{}{}
\apptocmd{\thmt@rst@storecounters}{\Hy@RestoreLastskip}{}{}
\makeatother

\crefname{property}{Property}{Properties}
\creflabelformat{property}{#2E#1#3}
\crefrangelabelformat{property}{#3E#1#4~--~#5E#2#6}

\crefname{assumption}{Assumption}{Assumptions}
\creflabelformat{assumption}{#2T#1#3}
\crefrangelabelformat{assumption}{#3T#1#4~--~#5T#2#6}

\hideLIPIcs
\nolinenumbers

\title{A Simple Partially Embedded Planarity Test Based on Vertex-Addition}
\funding{Funded by the Deutsche Forschungsgemeinschaft (German Research Foundation, DFG) grant RU-1903/3-1 and by the Vienna Science and Technology Fund (WWTF) [10.47379/ICT22029].}

\author{Simon D. Fink}{Algorithms and Complexity Group, Technische Universität Wien, Austria}{sfink@ac.tuwien.ac.at}{https://orcid.org/0000-0002-2754-1195}{}

\author{Ignaz Rutter}{Faculty of Informatics and Mathematics, University of Passau, Germany}{rutter@fim.uni-passau.de}{https://orcid.org/0000-0002-3794-4406}{}

\author{Sandhya {T. P.}}{Stockholm University, Department of Mathematics, Sweden}{thekkumpadan@math.su.se}{https://orcid.org/0000-0002-7745-3935}{}

\authorrunning{S.\,D. Fink and I. Rutter and Sandhya {T.\,P.}}

\Copyright{Simon D.\ Fink, Ignaz Rutter, and Sandhya {T.\,P.}}



\begin{document}

\maketitle
\thispagestyle{empty}

\section*{Abstract}
In the \peplan problem, we are given a graph $G$ together with a topological drawing of a subgraph $H$ of $G$.
The task is to decide whether the drawing can be extended to a drawing of the whole graph such that no two edges cross.
Angelini et al.\ gave a linear-time algorithm for solving this problem in 2010~\cite{abf-tpo-10,abf-tpo-15}.
While their paper constitutes a significant result, the algorithm described therein is highly complex:
it uses several layers of decompositions according to connectivity of both $G$ and $H$, its description spans more than 30 pages, and can hardly be considered implementable.
We give an independent linear-time algorithm that works along the well-known vertex-addition planarity test by Booth and Lueker~\cite{boo-pta-75,bl-tft-76}.
We modify the PC-tree as underlying data structure used for representing all planar drawing possibilities in a natural way to also respect the restrictions given by the prescribed drawing of the subgraph~$H$.
The testing algorithm and its proof of correctness only require small adaptations from the comparatively much simpler generic planarity test, of which several implementations exist.
If the test succeeds, an embedding can be constructed using the same approaches that are used for the generic planarity test.

\newpage
\setcounter{page}{1}

\section{Introduction}
\label{sec:intro}

In the \emph{partial representation extension problem}, the input consists of a graph $G$ and a representation~$\mathcal H$ of a subgraph~$H \subseteq G$.
The question is whether there exists a representation~$\mathcal G$ of $G$ whose restriction to~$H$ coincides with~$\mathcal H$.
The complexity of the problem strongly varies with the type of representation that is considered.
For planar straight-line drawings, the problem was shown to be NP-hard~\cite{pat-oea-06}, and in fact recently turned out to be $\exists\mathbb R$-complete~\cite{lmm-tco-18}.
On the other hand, in recent years a plethora of algorithmic and complexity results have been established for various classes of representations.
For example, Klav\'ik et al., who coined the term \emph{partial representation extension}, solved the problem for interval representations~\cite{kkv-epr-11} in quadratic time, which they later improved to linear \cite{kko-epr-16}.
Shortly afterwards, Angelini et al.~\cite{abf-tpo-10, abf-tpo-15} gave a linear-time algorithm for extending planar topological drawings.
Since then, the problem has been studied for a variety of different types of intersection representations, e.g., proper and unit interval graphs~\cite{kko-epr-17}, permutation graphs~\cite{kkkw-epr-12}, circle graphs~\cite{cfk-epr-19}, contact representations of geometric objects~\cite{cdk-cro-14}, trapezoid graphs~\cite{kw-epr-17} and rectangular duals~\cite{ckk-epr-21}.
In the context of drawings, the problem has also been studied for orthogonal drawings~\cite{ars-epo-21}.
In this work we focus on the case of planar topological drawings, for which Angelini et al.~\cite{abf-tpo-15} gave a linear-time algorithm.

\subparagraph*{Prior Work.}
In their paper, Angelini et al.~\cite{abf-tpo-15} first give a combinatorial characterization for yes-instances of \peplan.
They show that it is necessary and sufficient for a yes-instance to respect both the cyclic edge orders around vertices and the relative positions of different connected components defined by $\mathcal H$.
In particular, for a biconnected graph~$G$, these ``compatibility constraints'' set out by $\mathcal H$ can be individually verified on the nodes of the so-called SPQR-tree of~$G$ \cite{pat-pta-13,bt-olp-96}, which models the triconnected components as well as all planar embeddings of~$G$.
This procedure can also be used to individually test each biconnected component of a connected graph.
However, it is also necessary to verify certain compatibility constraints between the different blocks.
Similarly, the authors show that in the disconnected case, \peplan can be solved by testing each connected component and verifying the compatibility of the relative positions for the different components.
This characterization leads to a polynomial-time algorithm by progressively decomposing the input graph into its connected, biconnected, and triconnected
components, while also verifying the compatibility constraints at each step.

In order to improve the running time to linear, Angelini et al. first give a dynamic programming algorithm that solves \peplan for biconnected graphs.
This algorithm uses complex subprocedures that handle more restricted subcases.
For the connected and disconnected cases, they then show that the additional compatibility constraints can also be tested in linear time.
While their work constitutes a significant result, the algorithm described therein is highly complex, its description spans more than 30 pages.
Even when the graph $H$ that comes with a fixed drawing is connected; i.e., the embedding is uniquely determined by the rotation system of~$G$, it uses a decomposition of $G$ first into its biconnected, and then into its triconnected components.
When $H$ is not connected, these algorithms are applied for each face of $H$.
The resulting algorithm is thus highly technical and relies on a large number of non-trivial subprocedures and data structures.
It is therefore not surprising that no implementation is available to date.

\subparagraph*{Testing Planarity via Vertex Addition.}
In this paper, we propose an alternative solution for the problem that, unlike the work of Angelini et al., does not rely on dynamic programming on graph decompositions such as the SPQR-tree.
Instead, we straightforwardly extend the well-known linear-time vertex-addition planarity test of Booth and Lueker~\cite{boo-pta-75,bl-tft-76}.
This test incrementally inserts the vertices of a biconnected graph in an order that ensures that both the already inserted vertices as well as the not-yet inserted ones, respectively, induce a connected subgraph.
At every step, all edges from inserted to not-yet inserted vertices thus have to lie on the outer face of the already inserted subgraph.
One only cares about the possible cyclic orders of such edges along the outer face, disregarding the different planar embeddings that yield these orders.
To insert the next vertex, all edges connecting it to already-inserted vertices need to be consecutive on the outer face, as otherwise further edges to not-yet inserted vertices would be enclosed.
The algorithm uses a PC-tree\footnote{A simplified, cyclic version of the PQ-trees originally used by Booth and Lueker~\cite{boo-pta-75,bl-tft-76}, which only represent linear orders.} to succinctly represent the set of possible edge orders on the outer face.
The PC-tree is a data structure that represents sets of cyclic orders of its leaves, using two different types of inner nodes to model the possibilities for rearranging its leaves.
Its main operation is to restrict the set of represented orders to only those orders that have a certain subset of leaves cyclically consecutive (i.e., uninterrupted by leaves from the complement subset).
After making a set of edges to already-inserted vertices consecutive, they are removed from the represented orders and replaced by the set of the remaining edges of the inserted vertex (i.e., those to not-yet-inserted vertices).
The algorithm then proceeds to the next not-yet-inserted vertex.

\subparagraph*{Augmented PC-trees.}
The literature contains several examples where this ordinary planarity test is transferred to new settings by adapting its PC-trees to represent further constraints.
One example is the \textsc{Level Planarity} problem~\cite{dbn-hap-88}, where vertices have to be drawn on predefined horizontal lines.
Leipert et al.~\cite{jlm-lpt-98} annotate PC-trees with information on the height of components and the space available within faces to test this problem; see also \cite[Chapter 5]{brue-pvf-21}.
Similarly, in the \textsc{Constrained Level Planarity} problem~\cite{br-pac-17} where vertices have to satisfy some predefined ordering constraints, these constraints can also be carried over into the PC-tree; see \cite[Figure 5]{br-pac-17}.
Another example is \textsc{Simultaneous Planarity}~\cite{hjl-tsp-13}, where we seek planar drawings of two graphs that share some vertices and edges such that any shared vertex or edge is represented by the same point or curve, respectively.
Haeupler et al. show that for restricted variants of this problem, the standard planarity test can be augmented to process both graphs at the same time while using synchronization constraints on PC-tree nodes to ensure that a simultaneous embedding is obtained~\cite{hjl-tsp-13}.

\begin{algorithm}[tp]
  $v_1, \ldots, v_n\leftarrow\texttt{st-Order}(G)$\;
  $T_1\leftarrow$ \textcolor{orange}{constrained} P-node with $\deg(v_1)$ leaves $E(v_1)$ \textcolor{orange}{copying constraints of $v_1$}\;
  \For{$i$ in $1, \ldots, n-2\textcolor{orange}{, n-1}$}{
    $\ell\leftarrow$ new leaf;\qquad
    $F\leftarrow$ edges between $v_{i+1}$ and $\{v_1,\ldots,v_i\}$\;
    $S'\leftarrow$ \textcolor{orange}{constrained} P-node with leaves $(E(v_{i+1})\setminus F)\cup\{\ell\}$ \textcolor{orange}{copying constr.\ of $v_{i+1}$}\;
    $T'\leftarrow T_i$ updated to have $F$ consecutive and with the subtree of $F$ replaced by $\ell$\;
    $T_{i+1}\leftarrow S'$ merged with $T'$ at $\ell$\;
    \color{orange}
    \uIf{empty intersection between rotations of $v_{i+1}$ and admissible orders of $T'$\label{line:intersect}}{
      \Return \texttt{false};}
    \ElseIf{constraints of $v_{i+1}$ fix flip of C-node $\mu$ in $T'$}{
      fix flip of respective C-node $\mu'$ in $T_{i+1}$;\label{line:fix-flip}}
  }
  \Return \texttt{true} if no PC-tree update returned the null tree, and \texttt{false} otherwise\;
  \caption{
    High-level algorithm for testing a biconnected graph $G$ for \peplan. Our additions to the plain planarity test are shown in \textcolor{orange}{orange}.
    See \Cref{sec:peplan-bicon} for the full exposition of both algorithms.
  }
  \label{alg:hl-bicon-plan}
\end{algorithm}

\subparagraph*{Our approach.}
Similarly, the core of our approach for testing \peplan is a modification of the PC-tree as underlying data structure of the ordinary planarity test.
Our modification allows it to additionally handle the constraints that stem from the partial drawing.
In addition to the changes of the underlying data structure, we only need to make very small adaptations in the planarity testing algorithm itself; the high-level comparsion in \Cref{alg:hl-bicon-plan} nicely highlights this.
As our main contribution, we obtain a strongly simplified algorithm that relies on one depth-first search together with a single non-trivial data structure\footnote{In contrast, the algorithm by Angelini et al.\ uses ``several auxiliary data structures, namely block-cutvertex trees, SPQR-trees, enriched block-cutvertex trees, block-face trees, component-face trees, and vertex-face incidence graphs''~\cite{abf-tpo-15} and requires dynamic programming on these data structures to obtain a solution in linear time.}.
To further underline the practical advantages of our algorithm, we give concise, implementation-level pseudo-code for it.
The description of our core algorithm is less than half as long as the description of the algorithm by Angelini et al.\ \cite{abf-tpo-15}, while at the same time being far less technically involved.
\iftoggle{long}{
  The algorithmic description is complemented by an detailed proof of its correctness.
As this is not strictly necessary for understanding or implementing our algorithm, we first present the algorithm itself in \Cref{sec:peplan} and defer formal proofs to after the summary in \Cref{sec:summary} for interested readers.
}{
  Detailed proofs of correctness complementing our algorithmic description can be found in the full version \cite{}.
  As those are not strictly necessary for understanding or implementing our algorithm, we focus on the algorithm itself in this short version.
}

\section{Preliminaries}
\label{sec:preliminaries}
\newcommand{\mypar}[1]{\subparagraph*{#1}}

\mypar{Circular Orderings.}
Let $X$ be a ground set.
Two linear orders~$\alpha$, $\beta$ of $X$ are \emph{equivalent}, denoted by~$\alpha \sim \beta$, if there exist linear orders~$\alpha_1,\alpha_2$ such that~$\alpha=\alpha_1 \alpha_2$ and~$\beta = \alpha_2\alpha_1$.
That is, two orders are equivalent if one is a rotation of the other.
A \emph{cyclic order}~$\sigma$ is an equivalence class of~$\sim$.
Given a linear order~$\alpha$, we write~$[\alpha]:= \{ \beta \mid \alpha \sim \beta\}$ for the corresponding cyclic order.
For an order~$\sigma$, we use $\overline{\sigma}$ to denote its reversal.

Let~$A$ be a subset of $X$ with $\emptyset \neq A \subseteq X$, let~$\alpha$ be a linear order of $X$ and let~$\sigma$ be a cyclic order of $X$.
The set~$A$ is \emph{consecutive} in~$\alpha$ if $\alpha = \alpha_1\alpha_2\alpha_3$, such that~$\alpha_2$ is a linear order of~$A$ and~$\alpha_1\alpha_3$ is a linear order of the elements in~$X \setminus A$.
The set~$A$ is \emph{consecutive} in a cyclic order~$\sigma$ if there exists a linear order~$\beta \in \sigma$ such that~$A$ is consecutive in~$\beta$.
If the superset $X$ of set $A$ is clear from the context, we write $A^c$ for $X\setminus A$.
Now let~$\sigma$ be a cyclic order such that~$A$ is consecutive in $\sigma$ and let~$a \notin X$.
We denote by~$\sigma[A]$ the cyclic order of~$A$ that is obtained from~$\sigma$ by removing the elements of~$A^c$.
We denote by~$\sigma[A \to a]$ the cyclic order of~$A^c \cup \{a\}$ obtained from~$\sigma$ by replacing the elements of~$A$ with the single element $a$.

Let~$\sigma,\tau$ be cyclic orders of $X_1,X_2$ with~$X_1 \cap X_2 = \{\ell\}$.
The \emph{merge} of~$\sigma$ and~$\tau$, denoted by~$\sigma \otimes_\ell \tau$ is the cyclic order that is obtained by merging the two orders at~$\ell$.
More precisely, let~$\sigma = [\alpha\ell]$ and~$\tau = [\ell\beta]$, then~$\sigma \otimes_\ell \tau = [\alpha\beta]$.
Note that this yields $\sigma = \setto{(\merge{\sigma}{\tau})}{X_2}$.
If one side only contains the single element~$\ell$, the merge will effectively remove $\ell$ from the other side, i.e., $[\alpha\ell] \otimes_\ell [\ell] = [\alpha]$.

\mypar{Planar Drawings and Embeddings.}
A drawing of a graph $G=(V,E)$ maps each vertex to a distinct point in~$\mathbb R^2$ and each edge to a Jordan arc in $\mathbb R^2$ that connects its two endpoints in such a way that the points of vertices are only contained in arcs for which they are an endpoint.
A drawing is \emph{planar} if no two edges share an interior point.
A planar drawing~$\Gamma$ partitions the remainder of~$\mathbb R^2$ into \emph{faces}; the connected components of~$\mathbb R^2 \setminus \Gamma$.
The single unbounded face is called the \emph{outer face}.
Two planar drawings~$\Gamma_1,\Gamma_2$ are \emph{equivalent} if there exists an ambient isotopy that transforms~$\Gamma_1$ into~$\Gamma_2$, i.e., there exists a continuous map~$F \colon \mathbb R^2 \times [0,1] \to \mathbb R^2$ where each of the maps~$F_t(x) := F(x,t)$ is a homeomorphism of~$\mathbb R^2$ such that~$F_0$ is the identity and~$F_1$ maps~$\Gamma_1$ to~$\Gamma_2$.
An equivalence class of planar drawings is called an \emph{embedding}.
For a connected graph~$G$ an embedding~$\mathcal E$ is described by a \emph{rotation system}, i.e., the cyclic order of the edges $E(v)$ around each vertex $v$, together with a designated outer face.
If~$G$ is not connected, a rotation system only describes the embedding of each connected component and we need additional information about their \emph{relative position}, i.e., within which face each component is embedded.
This can be the outer face or a face of another connected component; in the latter case the one component is nested within the other.
We sometimes consider drawings or embeddings on a disk instead of the unbounded plane where some part of the graph has to lie on the boundary of this disk.
We will highlight this by using $\Gamma^\oplus$ and $\mathcal{E}^\oplus$ instead of $\Gamma$ and $\mathcal{E}$, respectively.

\mypar{PC-trees.}\label{sec:pc-trees}
PC-trees, introduced by Shih, Hsu and McConnell\cite{sh-anp-99,hm-pta-03,hm-ptp-04}, represent cyclic orders of a base set $X$.
A PC-tree $T$ is an unrooted tree with leaves $X$ and inner nodes of degree at least 3, each of which is either a \emph{P-node} or a \emph{C-node}.
See \Cref{fig:pc-ops} for examples,
where P-nodes are denoted by small disks and C-nodes by larger double circles.
While the edges incident to a P-node can be rearranged without any restriction, the edges incident to a C-node come with a cyclic order that is fixed up to reversal.
Any embedding of a PC-tree $T$ that respects this constraint induces a cyclic order of its leaves that we call an \emph{admissible order} of $T$.
The set of all admissible orders of $T$ is denoted by $\Ord(T)$.
We define the \emph{null-tree} to be the PC-tree with $\Ord(T)=\emptyset$.
We refer to the set of all leaves of $T$ as $L(T)=X$.

Let~$\emptyset \ne A \subsetneq L(T)$.
We now want to classify when $A$ is consecutive in every admissible order of $T$, that is whether $T$ allows for any cyclic order that has $A$ non-consecutive.
An edge~$e$ of $T$ is \emph{consistent} with $A$ if one of the two subtrees obtained by removing $e$ contains only leaves from~$A$.
We denote by $A(e) \subseteq A$ the leaves of this subtree.
For two consistent edges $e,e'$ of $T$, we define~$e \prec e'$ if $A(e) \subseteq A(e')$.
This gives a partial order on the consistent edges of $A$.
We denote by $E(T,A)$ the set that contains all maximal elements of this partial order; see~\Cref{fig:pc-eta}.
We call $A$ \emph{consecutive} (with respect to $T$) if $E(T,A)$ either consists of a single edge or is a consecutive set of edges around a $C$-node.
Observe that $A$ is consecutive if an only if $A$ is consecutive in every order~$\sigma \in \Ord(T)$.

\begin{figure}[t]
  \centering
  \begin{subfigure}{\linewidth}
    \includegraphics[width=\linewidth]{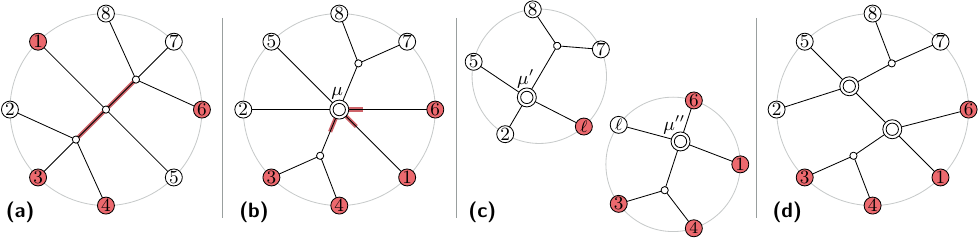}
    \phantomsubcaption\label{fig:pc-init}
    \phantomsubcaption\label{fig:pc-update}\label{fig:pc-eta}
    \phantomsubcaption\label{fig:pc-split}
    \phantomsubcaption\label{fig:pc-merge}
    \vspace{-.7cm}
  \end{subfigure}
  \caption{
    \textbf{(a)} A PC-tree $T$ on the set $L(T)=\{1,\ldots,8\}$ with only P-nodes as inner nodes. The red leaves belong to a set $A$ and their terminal path (described in \iftoggle{long}{\Cref{lem:ccpc-update,sec:pc-tree-details}}{\Cref{lem:ccpc-update}}) is highlighted in red.
    \textbf{(b)} The PC-tree $\update{T}{A}$ ensuring that the edges in $A$ are consecutive. Here, $E(T,A)$ consists of the three edges marked in red incident to C-node $\mu$.
    \textbf{(c)} The PC-trees $T'=\usplit{T}{A}$ (left) and $T''=\usplit{T}{A}[A^c]$ (right) showing the split of the previous tree and how $\mu$ is split into two parts $\mu'$ and $\mu''$.
    \textbf{(d)} The PC-tree $\merge{T'}{T''}$ showing the merge of the previous two trees.
  }
  \label{fig:pc-ops}
\end{figure}

\begin{statelater}{pcOps}
\label{sec:pc-tree-details}
We now describe the basic operations of PC-trees, which we already summarized in the preliminaries in \Cref{sec:preliminaries}, in more detail.

\begin{description}
\item[\Merge]
  Let~$T_1,T_2$ be two PC-trees whose respective leaf sets have size at least 2 and that share exactly one leaf~$\ell$; see \Cref{fig:pc-split}.
  The {\Merge} of $T_1,T_2$, denoted as $\merge{T_1}{T_2}$, is the PC-tree $T$ obtained by identifying the two copies of~$\ell$ in~$T_1$ and~$T_2$
  and smoothing the resulting degree-2 node~$\ell$ into an edge~$xy$, where $x,y$ are the two neighbors of~$\ell$ in~$T_1$ and~$T_2$, respectively; see \Cref{fig:pc-merge}.
Formally, we have $\Ord(T) = \{ \sigma_1 \otimes_{\ell} \sigma_2 \mid \sigma_i \in \Ord(T_i)$ for $i\in\{1,2\}\}$.
  The orders $\sigma_1 \in \Ord(T_1)$ and $\sigma_2 \in \Ord(T_2)$ corresponding to a $\sigma\in\Ord(T)$ can be obtained by undoing the merge that created $T$ from $T_1$ and $T_2$ while maintaining the embedding of $T$ that corresponds to $\sigma$.
  Observe that the leaves of each input tree are consecutive in the tree resulting from the merge, i.e., $L(T_i)$ is consecutive with respect to $T$ for $i=1,2$.
  Furthermore, any PC-tree can be obtained by merging trees with a single inner node.
We also extend the definition of \Merge to the case where one tree, say $T_2$, consists only of a single leaf $\ell$.
  In this case, analogously to cyclic orders, we simply remove $\ell$ from $T_1$.

\item[\Split]
  Let~$T$ be a PC-tree and let set~$A$ with $\emptyset \neq A \subsetneq L(T)$ be consecutive with respect to~$T$.
  The operation {\Split} separates $T$ into two new PC-trees $T'$ and~$T''$ representing the admissible orders of $A^c = L(T)\setminus A$ and $A$, respectively, in~$T$; see also \Cref{fig:pc-split}.
  The two trees have leaves $L(T')=A^c\cup\{a\}$ and $L(T'')=A\cup\{a\}$, where $a\notin L(T)$ is a new leaf that represents the position of the split-off subtree in each of the resulting halves.
  The PC-tree~$T'$ is obtained by replacing the edges in $E(T,A)$ by the single new leaf~$a$ and removing the subtrees containing the leaves in $A$.
  Symmetrically, the PC-tree $T''$ is obtained by replacing the edges in $E(T,A^c)$ with $a$ and removing the subtrees containing the leaves in $A^c$.
  This yields trees $T', T''$ with
  $\Ord(T') = \{ \sigma[A \to a] \mid \sigma \in \Ord(T)\}$ and
  $\Ord(T'') = \{ \sigma[A^c \to a] \mid \sigma \in \Ord(T)\}$.
  To refer to one of the resulting trees, we will borrow this notation and write $\setto{T}{A}[a]$ and $\setto{T}{A^c}[a]$ for the trees $T'$ and $T''$, respectively.

\item[\Update] Let~$T$ be a PC-tree and let $A \subseteq L(T)$ be a set of leaves.
  The operation {\Update}, denoted as $\update{T}{A}$, produces a new PC-tree $T'$ with~$\Ord(T') = \{ \sigma \in \Ord(T) \mid A$ is consecutive in $\sigma\}$.
  We also call the set $A$ a \emph{restriction} (of the admissible orders of $T$ to those where $A$ is consecutive).
  The procedure has the property that the leaf set $A$ is consecutive with respect to the resulting tree $T'$.
  We call a restriction \emph{impossible} if there is no admissible order of $L$ where the leaves in $R$ are consecutive, i.e., $\update{T}{A}$ is the null-tree.
  Note that leaf sets with $|A|\in\{0,1,|L(T)|-1, |L(T)|\}$ are always consecutive and thus do not require changes to $T$.
  Otherwise, the required changes are made by the following steps initially described by Hsu and McConnell \cite{hm-pta-03,hm-ptp-04}.
  We use the step numbering from \cite{fpr-eco-21} to stay consistent with the pseudo-code shown in \Cref{alg:rcpc-update}, but combine the first and last two steps respectively as their distinction is not relevant for this work.

  \begin{figure}
    \centering
    \begin{subfigure}[]{0.3\textwidth}
      \centering
      \includegraphics[page=1]{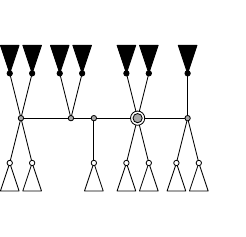}
      \caption{After step \ref{pc-upd-reorder}.}
      \label{fig:pc-upd-A}
    \end{subfigure}
    \hfill
    \begin{subfigure}[]{0.3\textwidth}
      \centering
      \includegraphics[page=2]{fig/update_step}
      \caption{After step \ref{pc-upd-split}.}
      \label{fig:pc-upd-B}
    \end{subfigure}
    \hfill
    \begin{subfigure}[]{0.3\textwidth}
      \centering
      \includegraphics[page=3]{fig/update_step}
      \caption{Final PC-tree.}
      \label{fig:pc-upd-C}
    \end{subfigure}
    \caption{Visualization of the updates to the terminal path made to ensure a set of leaves is consecutive. The full subtrees with only leaves that should be made consecutive are shown in black, empty subtrees are shown in white. The terminal path is the horizontal line with gray nodes.}
    \label{fig:pc-upd-steps}
  \end{figure}

  \begin{description}
  \item[\negphantom{1\&}1\&2] \namedlabel{pc-upd-labeling}{1\&2}
    Determine the edges of $T$ that are consistent with neither~$A$ nor~$L(T) \setminus A$; see \Cref{fig:pc-init,fig:pc-upd-steps}.
    If these edges do not form a path, the so-called \emph{terminal path}, then $T$ does not represent any cyclic order where~$A$ is consecutive, and we return the null-tree.

  \item[3] \namedlabel{pc-upd-reorder}{3}
    Reorder the edges around the nodes on the terminal path so that
    all subtrees that have all their leaves in $A$, which we will call \emph{full}, lie on one side and
    all subtrees with leaves in $L(T) \setminus A$, which we will call \emph{empty}, lie on the other side; see \Cref{fig:pc-upd-A}.
    If this is not possible, then this is due to a C-node around which edges consistent with~$A$ and edges consistent with $L(T) \setminus A$ alternate.
    It follows that $T$ does not represent a cyclic order where $A$ is consecutive, and we return the null-tree.

  \item[4] \namedlabel{pc-upd-split}{4}
    Split each P-node~$\mu$ on the terminal path twice, once to move all edges to full subtrees adjacent to $\mu$ to a new P-node~$\mu^F$
    and a second time to move all edges to empty subtrees to a new P-node~$\mu^E$.
    Add edges to the new nodes, making the remainder of $\mu$ adjacent to $\mu^F,\mu^E,$ and up to two edges of the terminal path; see \Cref{fig:pc-upd-B}.
    Convert this remaining part into a C-node $\mu^M$ and choose its embedding such that the two terminal edges are not adjacent to each other if it has degree 4,
    and flip it so that all full subtrees again lie on the same side of the terminal path.

  \item[\negphantom{5\&}5\&6] \namedlabel{pc-upd-contract}{5\&6}
    Contract all nodes of the terminal path (which are now all C-nodes) into a single, central C-node.
    Finally, smooth degree-2 vertices and remove degree-1 vertices that are not leaves of the original tree $T$; see \Cref{fig:pc-upd-C}.
  \end{description}

  Hsu and McConnell~\cite{hm-pta-03, hm-ptp-04} show that each of these steps can be implemented to run in time proportional to $|A|$ plus the length of the terminal path.
  This leads to amortized linear time in the size of the set $A$ as all terminal edges disappear through the update.

\item[\Intersect] Let~$T_1,T_2$ be two PC-trees with the same set of leaves.
  Then the operation {\Intersect} produces a new PC-tree $T$ with~$\Ord(T) = \Ord(T_1) \cap \Ord(T_2)$.
  Booth~\cite{boo-pta-75} describes a linear-time algorithm for computing the intersection of two PQ-trees; the same algorithm can be applied for PC-trees \cite{pfr-alt-20}.
  The general idea is to convert $T_1$ into a set of consecutivity constraints and to then update $T_2$ with these sets so that it represents only the orders that are represented by both trees.
  The key to achieve linear running time is to contract maximal subtrees that are already consecutive in both trees into single nodes.
\end{description}
\end{statelater}

\enlargethispage{1em}
  We will need the following well-known basic operations of PC-trees; see also \Cref{fig:pc-ops}.
  We only shortly define them here and give more in-depth explanations in \iftoggle{long}{\Cref{sec:pc-tree-details}}{the full version \cite{}}.
  We also point the interested reader to \cite{fpr-eco-21}, where PC-trees are described extensively in terms of pseudo-code and are also evaluated via a practical, open-source C++ implementation.

  \begin{description}
  \item[\Merge]
  Let~$T_1,T_2$ be two PC-trees with $|L(T_1)|, |L(T_2)|\geq2$ and $|L(T_1)| \cap |L(T_2)| = \ell$; see \Cref{fig:pc-split}.
  The {\Merge} of $T_1,T_2$, denoted as $\merge{T_1}{T_2}$, is the PC-tree $T$ with $\Ord(T) = \{ \sigma_1 \otimes_{\ell} \sigma_2 \mid \sigma_i \in \Ord(T_i)$ for $i\in\{1,2\}\}$
  obtained by identifying $\ell$ in $T_1, T_2$; see \Cref{fig:pc-merge}.

  \item[\Split]
  Let~$T$ be a PC-tree and let set~$A$ with $\emptyset \neq A \subsetneq L(T)$ be consecutive with respect to~$T$.
  The operation {\Split} separates $T$ into two new PC-trees $T'$ and~$T''$, with
  $\Ord(T') = \{ \sigma[A \to \ell] \mid \sigma \in \Ord(T)\}$ and
  $\Ord(T'') = \{ \sigma[A^c \to \ell] \mid \sigma \in \Ord(T)\}$, where $\ell\notin L(T)$ is a new leaf.
  Note that this may split a C-node $\mu\in T$ into two halves $\mu'\in T',\mu''\in T''$ as shown in \Cref{fig:pc-split}.
  We also write $\setto{T}{A}[\ell]$ for $T'$ and $\setto{T}{A^c}[\ell]$ for $T''$.

  \item[\Update]
  Let~$T$ be a PC-tree and let \emph{restriction} $A \subseteq L(T)$ be a set of leaves.
  The operation {\Update}, denoted as $\update{T}{A}$, produces a new PC-tree $T'$ with~$\Ord(T') = \{ \sigma \in \Ord(T) \mid A$ is consecutive in $\sigma\}$ in time linear in $|A|$ \cite{hm-pta-03,hm-ptp-04,fpr-eco-21}; see \Cref{fig:pc-update}.
  We call a restriction \emph{impossible} if there is no admissible order of $L$ where the leaves in $R$ are consecutive, i.e., $\update{T}{A}$ is the null-tree.

  \item[\Intersect]
  Let~$T_1,T_2$ be two PC-trees with the same leaf set, i.e., $L(T_1)=L(T_2)$.
  Operation {\Intersect} yields a new PC-tree $T$ with~$\Ord(T) = \Ord(T_1) \cap \Ord(T_2)$ in linear time~\cite{boo-pta-75,pfr-alt-20}.
  \end{description}

Note that, unlike for individual cyclic orders where splitting and merging are converse operations, the same does not always hold for PC-trees; see \Cref{fig:pc-merge} for an example.
Here, only certain orders $\sigma_1\in \Ord(T')$ and $\sigma_2\in \Ord(T'')$ can be merged to an admissible order of the original $T$.
We call such pair, that is where $\merge{\sigma_1}{\sigma_2}\in\Ord(T)$ holds, \emph{compatible}.
The following lemma shows that an incompatibility arises only from the two halves $\mu'\in T',\mu''\in T''$ a C-node $\mu\in T$ can be split into; see \Cref{fig:pc-split}.
See \iftoggle{long}{\Cref{sec:pc-tree-details}}{the full version \cite{}} for the full proof.

\begin{restatable}\restateref{lem:merge-splits}{lemma}{lemMergeSplits}
  \label{lem:merge-splits}
  Let set $A$ be consecutive in PC-tree $T$, let $T'=\setto{T}{A}$ and $T''=\setto{T}{A^c}$ and let $\sigma_1\in \Ord(T')$ and $\sigma_2\in \Ord(T'')$.
  If $E(T,A)$ is a single edge $e$, then $T=\merge{T'}{T''}$ and $\merge{\sigma_1}{\sigma_2}\in\Ord(T)$, that is any pair of $\sigma_1$ and $\sigma_2$ is compatible.
  Otherwise, $E(T,A)$ is a set of edges consecutive around a C-node $\mu$ of~$T$, and we have $T\neq\merge{T'}{T''}$ and $\Ord(T)\subsetneq\Ord(\merge{T'}{T''})$.
  In this case, $\sigma_1$ and $\sigma_2$ are compatible if and only if they induce the same flip of the split halves of $\mu$ in $T'$ and~$T''$.
  If $\sigma_1$ is not compatible with $\sigma_2$, it is instead compatible with $\overline{\sigma_2}$.
\end{restatable}
\begin{prooflater}{proofMergeSplits}
If $E(T,A)$ is a single edge $e$, tree $T$ is split by splitting $e$.
It can thus be reobtained by merging at $e$ again, that is $T=\merge{T'}{T''}$.
As the embeddings that~$\sigma_1$ and~$\sigma_2$ induce on~$T'$ and~$T''$, respectively, can also be joined at~$e$ to obtain an embedding of~$T$, we always have $\merge{\sigma_1}{\sigma_2}\in\Ord(T)$.

Otherwise, $E(T,A)$ is a set of edges consecutive around a C-node~$\mu$ of~$T$.
Splitting~$T$ into two trees $T', T''$ also splits~$\mu$ into two respective C-nodes $\mu',\mu''$, where~$\mu''$ is incident to the edges in $E(T,A)$ plus~$\ell$ and~$\mu'$ gets the remaining edges of~$\mu$ plus another copy of~$\ell$; see \Cref{fig:pc-merge}.
Merging $T'$ and $T''$ now does not yield $T$ again, as $\mu'$ and $\mu''$ are still separate C-nodes connected by the edge $\ell$ in $T^*=\merge{T'}{T''}$.
We have $\Ord(T^*)\supsetneq\Ord(T)$ as $\mu'$ and $\mu''$ can be flipped independently and thus in total allow four different orders for their incident edges in $T^*$, while $\mu$ in $T$ only allows two.
The embeddings $\sigma_1$ and $\sigma_2$ induce on $T'$ and $T''$ can thus only be merged to an embedding of $T$
if the rotations of $\mu'$ and $\mu''$ they induce can be merged to form an admissible rotation of $\mu$.
As the C-nodes have two admissible embeddings, either~$\sigma_1$ and $\sigma_2$ are compatible or $\sigma_1$ and $\overline{\sigma_2}$ are compatible.
\end{prooflater}

\section{Partially Embedded Planarity}\label{sec:partial}\label{sec:peplan}
\enlargethispage{1em}
Recall that for \peplan, we are given an instance $(G,H,\mathcal H)$ where $G$ is a graph with subgraph $H$ and $\mathcal H$ is an embedding of $H$.
We seek a planar embedding~$\mathcal G$ of $G$ whose restriction to~$H$ coincides with~$\mathcal H$.
In their solution, Angelini et al.~\cite[Lemma 3.9]{abf-tpo-15} ensure this condition by enforcing that
(i) around each vertex~$v$ of $H$ the cyclic order of the edges of $H$ is the same in~$\mathcal G$ and in~$\mathcal H$ and
(ii) for each (directed) facial cycle of~$\mathcal H$, the vertices of $H$ that are embedded left and right of it coincide in~$\mathcal G$ and in~$\mathcal H$.
It is the second condition, also referred to as having correct \emph{relative positions}, that is relatively complicated to handle efficiently.
We will show how to do this in the remainder of this section, which is organized as follows.
In \Cref{sec:cc-verts}, we show how the condition of having correct relative positions can be translated to color-coded constraints on the order of edges incident to individual vertices.
\Cref{sec:cc-pc} shows how such constraints can be represented using augmented PC-trees and how to update such PC-trees with new consecutivity restrictions.
In \Cref{sec:peplan-bicon}, we describe how substituting our augmented PC-trees in the Booth-Lueker planarity test for biconnected graphs can be used to test \peplan.
\Cref{sec:peplan-nonbic} lifts the biconnectivity requirement using the Haeupler-Tarjan generalization of the Booth-Lueker test to not-necessarily biconnected graphs.

\subsection{Color-Constrained Vertices}\label{sec:cc-verts}
Observe that, if $G$ is connected, the embedding of~$G$ and also the induced relative positioning of the connected components of $H$ can be expressed solely in terms of the rotation system of~$G$.
Angelini et al.~\cite[Theorem 4.14]{abf-tpo-15} give a straight-forward linear-time reduction from the case of a non-connected $G$ to solving the components of $G$ independently.
We can thus for now limit our attention to the case where~$G$ is connected.

\begin{figure}[t]
  \centering
  \includegraphics[page=3]{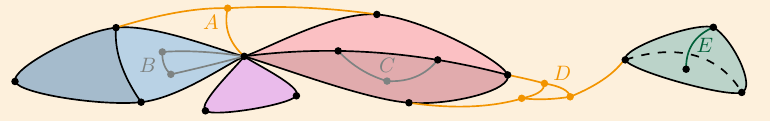}
  \caption{
    A positive instance of \peplan. The graph $H$ is shown with black edges.
    The graph $G-H$ is divided into 5 bridges $A,B,C,D$ and $E$.
    Bridges $B$ and $C$ have all their attachments in a single block and are thus not restricted to a face.
    Bridges $A, D,$ and $E$ have attachments in different blocks and are thus restricted to a face.
    Adding the dashed edge to $H$ as shown would turn the instance negative, as the attachments of $E$ would no longer share a face.
  }
  \label{fig:peplan-bridges}
\end{figure}

Now consider two distinct connected components $C_1, C_2$ of $H$ connected by a path $p$ whose inner vertices belong to $V(G)\setminus V(H)$.
The fact that~$p$ does not contain vertices of $H$ except for its endpoints already ensures that $C_1, C_2$ are embedded on the same side of any facial cycle of any other component of $\mathcal H$.
It remains to ensure that $C_2$ is embedded in the correct face of~$C_1$ and vice versa.
Thanks to the connectivity of $G$, ensuring this for each pair $C_1, C_2$ of components of $H$ connected by such a path $p$ is sufficient to ensure correct relative positions.
Note that not only $p$ has to be in a certain face of~$\mathcal H$, but this also applies to the whole connected component of $G - V(H)$ that contains $p$, which we call an $H$-bridge.
Formally, an \emph{$H$-bridge} $B$ is either a single edge $e \in E(G)\setminus E(H)$ with both end-vertices in $H$ or a connected component $B$ of $G - V(H)$; see \Cref{fig:peplan-bridges}.
The \emph{attachments} of a component $B$ are the vertices of $H$ whose removal disconnects $B$ from the remaining graph or, in the case of a single edge $B$, its endpoints.
Note that each $H$-bridge of $G$ has to lie in exactly one face of $\mathcal H$ as it contains no vertices of $H$.
If an $H$-bridge $B$ has attachments in at least two distinct blocks of $H$, then that face is uniquely determined; see \cite[Section 2.3]{abf-tpo-15} and \Cref{fig:peplan-bridges}.
Thus, we can color each edge $e\in E(G)\setminus E(H)$ with the unique face $f(e)$ of $\mathcal H$ in which $e$ must be embedded, or $e$ is uncolored as this face is arbitrary and we set $f(e)=\bot$.
Angelini et al.\ give a simple linear-time algorithm computing this coloring \cite[Lemma 2.2]{abf-tpo-15}.
To ensure correct relative positions, we now need to ensure that, at every vertex, all edges to incident $H$-bridges are embedded in the correct face of $\mathcal H$.
Furthermore, we need to respect the cyclic order of incident edges of $H$ as given in $\mathcal H$, where each angle between two consecutive edges corresponds to a face of $\mathcal H$.
We will express both requirements as constraints on vertex~rotations.

\enlargethispage{1em}
Let~$v$ be a vertex of~$G$.
A \emph{color-constraint}~$C_v=(F_v,R_v,U_v,\rho_v,f)$ for~$v$ partitions the edges incident to $v$ into a set~$F_v=E(v)\cap E(H)$ of \emph{fixed edges}, a set~$R_v=\{ e\in E(v)\setminus E(H) \mid f(e) \neq \bot \}$ of \emph{restricted edges}, and a set~$U_v=\{ e\in E(v)\setminus E(H) \mid f(e) = \bot \}$ of \emph{unrestricted edges}; see \Cref{fig:color-constraint}.
The fixed edges have a fixed counter-clockwise order~$\rho_v=\mathcal H(v)$, in which we want to insert the remaining edges to find a rotation for $v$.
We additionally constrain where the restricted edges can be inserted.
For each fixed edge $h$, $\mathcal H$ defines a face $f_v(h)$ which is incident to $h$ in counter-clockwise direction around $v$.
We call a pair of fixed edges~$(h_1,h_2)$ that are consecutive in~$\rho_v$ an \emph{angle}; see \Cref{fig:color-constraint}. A \emph{valid cyclic order}~$\sigma$ (i.e., one that \emph{satisfies} $C_v$) of the edges incident to~$v$ is obtained by arbitrarily inserting the restricted and unrestricted edges of~$v$ into~$\rho_v$ in such a way that, for each restricted edge~$e$ with (in counter-clockwise order) preceding fixed edge $h$, $f(e) = f_v(h)$ holds.
If $F_v=\emptyset$, any order of $R_v\cup U_v$ is valid.
\Cref{fig:color-constraint} shows an example color-constraint, where the restricted edge $e$ with label 5 and the fixed edge $h$ with label~4 have $f(e) = f_v(h)$, represented through purple color.
Observe that any planar embedding that satisfies all color-constraints satisfies both conditions~(i) and~(ii) from the beginning of this section, therefore yielding a solution to our problem.
Note that while we derive the color-constraint $C_v$ from $\mathcal H$, its definition does not rely on $\mathcal H$ as we can interpret $f$ as a coloring with arbitrary, opaque values.
\enlargethispage{1em}

\begin{figure}
  \begin{subfigure}{.3\textwidth}
    \centering
    \includegraphics[page=1]{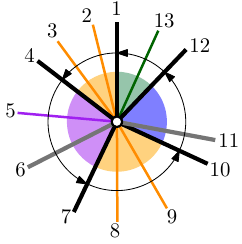}
    \vspace{-.7cm}\caption{\hspace*{\fill}}
  \end{subfigure}\hfill \begin{subfigure}{.3\textwidth}
    \centering
    \includegraphics[page=2]{fig/partial-pnode}
    \vspace{-.7cm}\caption{\hspace*{\fill}}
  \end{subfigure}\hfill \begin{subfigure}{.3\textwidth}
    \centering
    \includegraphics[page=3]{fig/partial-pnode}
    \vspace{-.7cm}\caption{\hspace*{\fill}}
  \end{subfigure}
  \caption{
    Example of a color-constraint.
    Fixed edges are fat and black, their fixed counter-clockwise order is indicated by arrows.
    The faces following fixed edges are indicated by colored angles.
    Each restricted edge is drawn using the respective color of its face, it may be inserted into an arbitrary angle of the same color.
    \textbf{(a)} and \textbf{(b)} show two different valid cyclic orders orders around~$v$.
    The order in \textbf{(c)} is not valid, since the orange edge $9$ is in the blue angle between $10$ and $12$.
  }
  \label{fig:color-constraint}
\end{figure}

\subsection{\Pcs}\label{sec:cc-pc}
As we want to respect the color-constraints of every graph vertex in our vertex-addition planarity test, we now introduce a variant of the PC-tree that can also encode these color-constraints as described in the previous section.
A \emph{\pc} has three different types of nodes:
\begin{itemize}
\item C-nodes, which behave exactly as in the case of normal PC-tree,
\item \emph{fixed C-nodes}, for which the order of their incident edges is completely fixed and may not even be reversed, and finally
\item \emph{color-constrained P-nodes}, which are P-nodes with a color-constraint as defined for graph vertices above.
\end{itemize}
\enlargethispage{1em}
When the context is clear, we refer to the latter simply as P-nodes.
Note that an ordinary P-node~$\mu$ is a special case of a color-constrained P-node with a color-constraint where~$F_\mu = R_\mu = \emptyset$.
Similarly, a color-constrained P-node with a color-constraint where $R_\mu = U_\mu = \emptyset$ is equivalent to a fixed~C-node.
As with usual PC-trees, choosing a valid cyclic order of the edges incident to each inner node of a \pc~$T$ determines a cyclic order of its leaf set~$L(T)$.
Therefore, $T$ also represents a set~$\Ord(T)$ of cyclic order of its leaves.

\begin{algorithm}[t]
  \tcp{Step 1, Labelling}
  \textsc{LabelNodes}(R)\tcp*[l]{unchanged \cite[Algorithm 1]{fpr-eco-21}}
  \tcp{Step 2, Terminal Path Enumeration}
  \textsc{EnumerateTP}()\tcp*[l]{unchanged \cite[Algorithm 3]{fpr-eco-21}}
  \vspace{.2cm}

  \tcp{Step 3, Node Flips}
  \textsc{CheckFlips}()\tcp*[l]{\textcolor{orange}{modifications described in \iftoggle{long}{\Cref{alg:rcpc-flips}}{\Cref{lem:ccpc-update}}}}
  \vspace{.2cm}

  \tcp{Step 4, Node Splitting}
  \For{P-node $\mu$ on the terminal path}{
    $F\leftarrow$ incident edges of $\mu$ leading to full subtrees\;
    $\mu',\,\mu^F  \leftarrow\textsc{Split}(\mu, F)$\tcp*[l]{\textcolor{orange}{modifications described in \iftoggle{long}{\Cref{sec:pc-update}}{\Cref{lem:ccpc-update}}}}
    $\overline{E}\leftarrow$ incident edges of $\mu'$ leading to terminal path nodes\;
    $\mu^E,\,\mu^M \leftarrow\textsc{Split}(\mu, \overline{E})$\tcp*[l]{\textcolor{orange}{modifications described in \iftoggle{long}{\Cref{sec:pc-update}}{\Cref{lem:ccpc-update}}}}

    \uIf{$\mu^M$ has degree 4}{$\rho\leftarrow$ cyclic order of $E(\mu^M)$ where terminal path edges are non-adjacent\;}
    \lElse{$\rho\leftarrow$ arbitrary cyclic order of $E(\mu^M)$}
    \color{orange}
    \uIf{$\mu^M$ allows order $\overline{\rho}$}{\color{black}make $\mu^M$ an ordinary C-node with default cyclic order $\rho$\;}
    \lElse{make $\mu^M$ a fixed C-node with fixed cyclic order $\rho$}
  }
  \vspace{.2cm}

  \tcp{Steps 5 \& 6, Terminal Path Contraction}
  \For{edge $e$ on the terminal path}{
    \textsc{Contract}(e)\tcp*[l]{unchanged \cite[Section 3.4]{fpr-eco-21}}
  }
  \textcolor{orange}{
\If{any C-node on the initial terminal path was fixed}{
    mark the C-node $\mu^C$ resulting from contractions as fixed\;
  }}

  \caption{
    Updating a \pc to make a set $R$ consecutive.
    Our changes over \cite{fpr-eco-21} are highlighted in \textcolor{orange}{orange}.
    Note that impossible restrictions are detected within steps 1--4 instead of 1--3 as for ordinary PC-trees.
  }
  \label{alg:rcpc-update}
\end{algorithm}

To use these trees in the vertex-addition planarity test, we extend the different operations of PC-trees to also respect color-constraints. The operations \Merge and \Split can be easily implemented in an analogous fashion to PC-trees.
The main operation of interest is \Update.
  We outline our modifications to the \Update operation from ordinary PC-trees (of which multiple practical implementations exist~\cite{fpr-eco-21}) in the proof sketch for the following lemma; see also \Cref{alg:rcpc-update}.
  \iftoggle{long}{\Cref{sec:pc-update}}{The full version \cite{}} contains the full description and proof of correctness.

  \begin{restatable}\restateref{lem:ccpc-update}{lemma}{lemCCPCUpdate}\label{lem:ccpc-update}
    For a \pc $T$ and a subset $L$ of its leaves, there exists a \pc $T'=\update{T}{A}$ with $\Ord(T')=\{\tau\in\Ord(T)\mid A \text{ is consecutive in } \tau\}$, which can be found in $O(|A|)$ time.
  \end{restatable}
  \begin{proof}[Sketch.]
  We briefly outline the linear-time \Update operation of ordinary PC-trees while also sketching which adaptions are neccesary to also apply it to \pcs; see also \Cref{alg:rcpc-update} for an overview over the algorithm and our changes.
  The core insight is that, for an (ordinary) PC-tree $T$ to have leaves $A$ consecutive, the set of edges that are consistent with neither~$A$ nor~$L(T) \setminus A$ need to form a path, the so-called \emph{terminal path}~\cite{hm-pta-03,hm-ptp-04}; see \Cref{fig:pc-init}.
  Otherwise $T$ does not represent any cyclic order where~$A$ is consecutive, and \Update returns the null-tree.

  After identifying the terminal path, the \Update performs three further steps~\cite{fpr-eco-21}.
  First, all edges around the nodes on the terminal path are reordered so that the subtrees with only leaves in $A$ (called \emph{full}) and those with only leaves in $L(T) \setminus A$ (called \emph{empty}) lie on different sides, separated by the up to two edges of the terminal path.
  If this is not possible, then this is due to a (fixed) C-node (or color-constrained P-node) around which edges consistent with~$A$ and edges consistent with $L(T) \setminus A$ must alternate.
  It follows that $T$ does not represent a cyclic order where $A$ is consecutive, and we return the null-tree.

  Next, split (i.e. create a new adjacent P-node and reassign some incident edges) each P-node~$\mu$ on the terminal path twice, once to move all edges to full subtrees adjacent to $\mu$ to a new P-node~$\mu^F$ and a second time to move all edges to empty subtrees to a new P-node~$\mu^E$\iftoggle{long}{; see \Cref{fig:pc-upd-steps,fig:p-update-steps} in \Cref{sec:pctree-details}}{}.
  Now, $\mu$ is adjacent to $\mu^F$, $\mu^E$ and up to two further nodes from the terminal path, in which case it is converted to a C-node to ensure that the full and empty subtrees always lie on different sides of the terminal path.
  When splitting color-constrained P-nodes, some additional care needs to be taken to correctly distribute the constraints across the new nodes and edges.
  In \iftoggle{long}{\Cref{sec:pc-update}}{the full version}, we show that the splitting procedure for ordinary P-nodes can easily be augmented to ensure this.
Lastly, all nodes of the terminal path are contracted into a single C-node which has all incident full subtrees consecutive, and which for \pcs needs to be fixed if any of the constituent C-nodes were fixed.
  \end{proof}

The biggest difference between usual and \pcs concerns the operation \Intersect.
In contrast to usual PC-trees, given two \pcs~$T_1,T_2$ on the same set of leaves, there generally does not exist a \pc~$T$ with~$\Ord(T) = \Ord(T_1) \cap \Ord(T_2)$.
For example, if both \pcs consist of a single P-node with different fixed edges, both fixed orders can be interleaved arbitrarily in the intersection, which cannot be represented using a \pc.
This is however not strictly needed for the planarity test.
Instead, we will here only need to test whether~$\Ord(T_1) \cap \Ord(T_2) \ne \emptyset$.
Conceptually, such a test can be performed using an approach similar to the intersection of ordinary PC-trees, additionally checking for each pair of inner nodes whether they allow for at least one common fixed order.
\iftoggle{long}{\Cref{lem:lt-intersect} in \Cref{sec:linear-time} shows}{In the full version, we show} that, in the context of our planarity test, this can be checked in linear time.

\begin{statelater}{pcUpdate}
\subsection{Color-constrained PC-tree Update Procedure}\label{sec:pc-update}
The \Update procedure on \pcs is based on the same steps as the update on ordinary PC-trees, although we need to make some modifications to account for the constraints.
These modifications, which we already outlined in the proof sketch for \Cref{lem:ccpc-update}, can be summarized as follows; see also \Cref{alg:rcpc-update}.
The labeling in step \ref{pc-upd-labeling} works the same as on regular PC-trees.
The reordering in step \ref{pc-upd-reorder} now also needs to respect fixed C-nodes and the order of fixed edges around P-nodes.
When splitting nodes in step \ref{pc-upd-split}, we need to correctly distribute the constraints across the new nodes and edges resulting from the split.
Especially, splitting a node may separate restricted edges from fixed edges and thus also from the angles they want to be embedded in, making the split impossible.
Lastly, the contractions in step \ref{pc-upd-contract} need to respect fixed C-nodes when merging C-nodes.
In the following, we describe the changes which need to be made to obtain the correct result in greater detail.
To implement these changes in amortized linear time, we will need to add some additional counters to the data structure, which we discuss in \Cref{sec:linear-time}.

\begin{algorithm}[t]
  \For{node $\mu$ on the terminal path}{
    \textcolor{orange}{
    \uIf{$\mu$ is a fixed P- or C-node}{
      $\rho\leftarrow$ cyclic order of fixed edges at $\mu$\;
    }}
    \uElseIf{$\mu$ is a C-node}{
      $\rho\leftarrow$ cyclic edge order fixed up to reversal at $\mu$\;
    }
    \Else{
      \textbf{continue} at next loop iteration\;
    }
    \If{edges leading to full subtrees are \emph{not} consecutive in $\rho$}{
      \textbf{abort} and report impossible restriction\;
    }
    \If{edges leading to full subtrees or terminal path nodes are \emph{not} consecutive in $\rho$}{
      \textbf{abort} and report impossible restriction\;
    }
  }
  \textcolor{orange}{
  \If{not all fixed P- and C-nodes have their fixed edges leading to full subtrees on the same side of the terminal path}{
      \textbf{abort} and report impossible restriction\;
  }}

  \caption{
    Operation \textsc{CheckFlips}().
  }
  \label{alg:rcpc-flips}
\end{algorithm}

In step \ref{pc-upd-reorder} we need to ensure that the full fixed edges (i.e. the incident fixed edges leading to full subtrees) are consecutive in the cyclic order of fixed edges.
Furthermore, we need to ensure that both of the at most two terminal fixed edges are directly adjacent to this block of full fixed edges, or adjacent to each other if the block is empty.
Finally, we need to check that all P-nodes and all fixed C-nodes have their fixed full edges on the same side of the terminal path.
If any of these checks fails, we abort and report an invalid restriction; see \Cref{alg:rcpc-flips}.
We will ensure that all restricted full edges are consecutive with the fixed ones in the next step.

\begin{figure}[t]
  \centering
  \includegraphics[width=\linewidth]{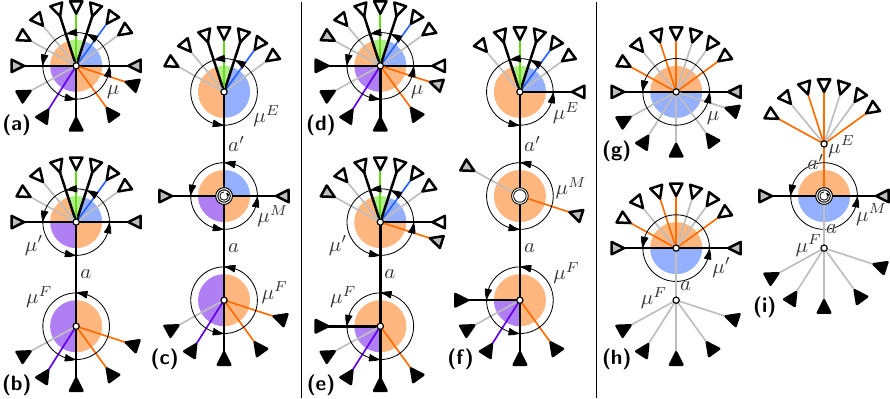}
  \caption{
    Splitting three different color-constrained P-nodes \textbf{(a)}, \textbf{(d)}, and \textbf{(g)}.
    The first step (\textbf{(b)}, \textbf{(e)}, and \textbf{(h)}) splits off the incident full (black) subtrees, the second step (\textbf{(c)}, \textbf{(f)}, and \textbf{(i)}) splits off the empty (white) subtrees.
    All splits up to \textbf{(f)} have fixed edges on both sides.
    The split in \textbf{(h)} splits off only unrestricted edges, while split \textbf{(i)} splits off unrestricted and restricted edges.
    Only the C-node created in \textbf{(f)} is not fixed, as the reversal of its shown rotation is also admissible.
  }
  \label{fig:p-update-steps}
\end{figure}

When splitting a color-constrained P-node with fixed edges in step \ref{pc-upd-split}, we also need to maintain the constraint information and especially ensure that the split parts still allow the same relative positions with regard to each other.
This especially means that the new edges connecting them need to have the right restrictions assigned.
We will describe how to do this when splitting off the full edges $F$ to $\mu^F$ (see \Cref{fig:pc-upd-B}), splitting off the empty ones from the resulting node $\mu'$ to $\mu^E$ works analogously\footnote{Note that in a linear-time implementation, we cannot process all incident empty edges. Instead splitting off the complement set, i.e., the terminal path edges, maintains the time bound.}.
We make a case distinction based on whether both $\mu^F$ and $\mu'$ receive at least one fixed edge; see \Cref{fig:p-update-steps}.
If this is the case~(\Cref{fig:p-update-steps}a--f), edge $a$ between $\mu^F$ and $\mu'$ is fixed at both ends and the orders of the fixed edges are set to
$\rho_{\mu'}=\rho_\mu[(F\cap F_\mu) \to a]$ and $\rho_{\mu^F}=\rho_\mu[(F^c\cap F_\mu) \to a]$.
All fixed edges retain their color, while $a$ is assigned the color that followed $F$ before the split at $\mu'$ and the color that preceded $F$ at $\mu^F$.
We need to check that all nodes still have at least one appropriate angle for every restricted edge, or abort and report an impossible restriction otherwise\footnote{As discussed in \Cref{sec:linear-time}, we can maintain counters tracking how many edge and angles of different colors are present at a node to check this.}.

If one of $\mu^F,\mu'$ received no fixed edges (\Cref{fig:p-update-steps}g--i), we assume without loss of generality that~$\mu'$ receives all fixed edges as the converse case works analogously.
Here, we set $\rho_{\mu'}=\rho_\mu$ and $\rho_{\mu^F}=\emptyset$, leaving the coloring of fixed edges as-is.
The restriction of edge $a$ is set according to which edges $\mu^F$ retained.
If there are restricted edges of more than one color at $\mu^F$, we abort and report an impossible restriction.
If there are restricted edges of exactly one color $c$ at $\mu^F$ (the analogous case for $\mu^E$ is shown in \Cref{fig:p-update-steps}i), we set the edge $a$ ($a'$ in \Cref{fig:p-update-steps}i) to be restricted to $c$ at~$\mu'$ and to be unrestricted at $\mu^F$.
If there are no restricted edges at $\mu^F$ (\Cref{fig:p-update-steps}h), $a$ is unrestricted at both its ends.

If the middle node $\mu^M$ resulting from the two splits has degree 4, we need to additionally restrict its order of incident edges such that the terminal path edges are non-adjacent, i.e., the full and empty nodes are on different sides of the terminal path.
Note that for a degree-4 node, there are at most two such admissible orders, which are the reverse of each other.
If both are allowed by $\mu^M$ (which currently is still a color-constrained P-node), we convert $\mu^M$ to an ordinary C-node with one of the two orders, otherwise to a fixed C-node with the single possible order; see~\Cref{fig:p-update-steps}.
For any degree-3 node, there are overall at most two orders, which are the reverse of each other; so can assume all these nodes to also be C-nodes.

Finally, we need to ensure that all middle nodes always have the empty and full subtrees, respectively, on the same side of the terminal path.
This is done by the contractions of adjacent C-nodes along the terminal path in step \ref{pc-upd-contract}.
The C-node resulting from this is fixed if and only if at least one of its constituent C-nodes was fixed.

\lemCCPCUpdate*\label{lem:ccpc-update*}
\begin{proof}
  We will show that the tree $T'$ we obtained by applying our modified \Update procedure as described above satisfies this condition.
  \Cref{lem:lt-update} in \Cref{sec:linear-time} details how to conduct the update in linear time.
  Note that $T'$ can be converted into an ordinary PC-tree, which we will refer to as $\Project(T')$, by converting all color-restricted P-nodes into ordinary P-nodes (dropping their color-constraints) and converting all fixed C-nodes into ordinary C-nodes (possibly now also allowing reversal of their orders).
  As our modified update makes the same changes to the tree structure as in the normal update operation we have $\Project(T')=\update{\Project(T)}{L}$ if the restriction is possible.
  As the projection only allows additional orders, we have $\Ord(T')\subseteq\Ord(\Project(T'))=\Ord(\update{\Project(T)}{L})$,
  in particular, $L$ is always consecutive in $T'$.

  To show the claimed equivalence of admissible orders,
  we first show that if $\tau\in\Ord(T)$ and $L$ is consecutive in $\tau$ then $\tau\in\Ord(T')$.
  Note that the restriction must be possible and $T'$ thus cannot be the null-tree, as an impossible restriction would imply that there is no $\tau\in\Ord(T)$ where $L$ is consecutive.
  As we have $\tau\in\Ord(\update{\Project(T)}{L})$ it is also $\tau\in\Ord(\Project(T'))$ due to the above equivalence and it remains to show that $\tau$ satisfies the color-constraints of $T'$.
  To do this, we will apply the changes made by the update procedure to $T$ while maintaining its embedding given by $\tau$ to obtain an admissible embedding of $T'$.
  First, observe that for all terminal nodes the incident full and empty subtrees with a fixed ordering are respectively consecutive and on different sides of the terminal path, as we would have otherwise returned a null-tree in step \ref{pc-upd-reorder}.
Now consider one of the two splits applied in step \ref{pc-upd-split}.
  Note that the split-off set $A$ needs to be consecutive in the embedding induced on the current P-node $\mu$ by order $\tau$ as otherwise $L$ would not be consecutive in $\tau$.
  Thus, the edges in $A$ can be reassigned to a new P-node $\mu'$ adjacent to $\mu$ in place of $A$ while maintaining the order of $A$.
  Note that if $A$ contains no fixed edges, all edges of $A$ were embedded in a single angle.
  If there are no fixed but restricted edges in $A$, they all need to have the same color which coincides with the color restricting the edge that replaces~$A$.
The conversion of $\mu^M$ into a C-node after the two splits can also be done while maintaining the embedding, as the only disallowed rotations of $\mu^M$ are those that do not have full and empty subtrees on different sides of the terminal path.
Finally, contracting all C-nodes on the terminal path cannot contradict the embedding as we already ensured that all full and empty subtrees are on the correct sides of the terminal path.

  Conversely, we need to show that if $\tau\in\Ord(T')$ then also $\tau\in\Ord(T)$.
  From our initial considerations it follows that we have $\tau\in\Ord(\Project(T'))=\Ord(\update{\Project(T)}{L})$ and it again remains to show that $\tau$ satisfies the color-constraints of $T$.
  To do this, we undo the changes made by the update procedure that turned $T$ into $T'$ while maintaining its embedding given by~$\tau$ to obtain an admissible embedding of $T$.
  As all changes only restrict the number of admissible embeddings, undoing them cannot turn an admissible embedding invalid.
\end{proof}
\end{statelater}

\subsection{Testing Biconnected Partially Embedded Graphs}\label{sec:peplan-bicon}
Having the underlying data structure in place, we now turn to our algorithm for testing \peplan.
In the following, we will describe the generic Booth \& Lueker planarity testing algorithm \cite{bl-tft-76} in two steps: first the underlying concepts, then the actual algorithm.
Directly after each step we highlight the changes that need to be made for testing partial instances.

To test planarity of graph $G=(V,E)$, we iteratively insert its vertices in a certain order, that is, in each step $i\in\{1,\dots,n\}$ we grow the set $V_i = \{v_1,\dots,v_i\} \subseteq V$ of already-inserted vertices.
At each step, we partition the edges of~$G$ into three types: \emph{Embedded edges} have both endpoints in~$V_i$, \emph{half-embedded} have exactly one endpoint in~$V_i$ and \emph{unembedded edges} have both endpoints in~$V \setminus V_i$.
When inserting vertex~$v_i$ into the graph, its incident unembedded edges become half-embedded and its half-embedded edges become embedded.
We denote by~$G_i$ the subgraph of~$G$ induced by~$V_i$, and by~$G_i^+$ the graph obtained from~$G_i$ by adding each half-embedded edge~$e=uv$ with~$u \in V_i$ as half-edge starting at~$u$; see \Cref{fig:bic-emb-graph}.

The central idea of the planarity test is to use a vertex order that has $G[V\setminus V_i]$ connected at each step $i\in\{1,\ldots,n\}$.
This ensures (by the Jordan curve theorem) that all half-embedded edges must be embedded in the same face of~$G_i^+$, without loss of generality, the outer face.
We will for now assume $G$ to be biconnected and use an st-ordering of its vertices.
For an \emph{st-ordering} $s=v_1,\ldots,v_n=t$ of $V$, the edge $st$ must be present in $G$ and each vertex except for~$s$ and~$t$ must have a neighbor which comes earlier in the ordering as well as a neighbor which comes later~\cite{et-cas-76}.
Especially, this ensures that both $G[V_i]$ and $G[V\setminus V_i]$ are connected. 

Observe that a planar embedding of~$G$ determines a planar embedding of each $G_i^+$.
Let $\Emb(G_i^+)$ denote the set of all planar embeddings of~$G_i^+$ with all half-edges on the outer face.
For an embedding $\mathcal E \in \Emb(G_i^+)$ of~$G_i^+$, let~$\Ord(\mathcal E)$ be its cyclic order of half-edges on the outer face; see \Cref{fig:bic-emb-graph}.
We define $\Ord(G_i^+)=\{\Ord(\mathcal E)\mid \mathcal E \in \Emb(G_i^+)\}$ to be the set of all such orders.

\subparagraph*{Neccesary Conditions for Partially Embedded Planarity.}
We now extend these notions to the partial setting.
Let $(G,H,\mathcal H)$ be a partially embedded graph where $G$ is biconnected.
Analogously to $G_i$, we set $H_i=H[V_i \cap V(H)]$, define $H_i^+$ to be the graph obtained from $H_i$ by adding all half-embedded fixed edges as half-edges, and define $\mathcal H_i^+$ as the restriction of $\mathcal H$ to $H_i^+$.
Let $\mathcal H_i^\oplus$ be a topological drawing of $\mathcal H_i^+$ inside a disk whose boundary visits the non-vertex endpoints of the half-edges according to their order on the outer face.

Consider an embedding $\mathcal G$ that is a solution for $(G,H,\mathcal H)$ and for $i\in\{1,\ldots,n\}$ a partial solution $\mathcal G_i^+$ that is a restriction of $\mathcal G$ to $G_i^+$.
Analogously to $\mathcal H_i^\oplus$, we define $\mathcal G_i^\oplus$ to be a topological disk-drawing of $\mathcal G_i^+$ with all half-edges ending at its boundary.
Note that each face of $\mathcal G_i^\oplus[H_i^+]$, which is the restriction of $\mathcal G_i^\oplus$ to $H_i^+$, corresponds to a distinct face of $\mathcal H_i^\oplus$.
The embedding $\mathcal G_i^+$ that is a partial solution satisfies the following three properties:
\begin{enumerate}[(E1)]
\item All half-edges lie on the outer face. \label[property]{prop:half-edges}
\item Drawing $\mathcal G_i^\oplus[H_i^+]$ coincides with $\mathcal H_i^\oplus$. \label[property]{prop:fixed-rotations}\item Each edge $e\in E(G_i^+)\setminus E(H)$ with $f(e)\neq\bot$ is embedded in a face of $\mathcal G_i^\oplus[H_i^+]$ that corresponds to $f(e)$ in $\mathcal H$. \label[property]{prop:colored-faces}
\end{enumerate}
Note that \cref{prop:half-edges} is the same as for the ordinary planarity test
and that any planar embedding $\mathcal G_n^+$ that satisfies \cref{prop:fixed-rotations} is a solution for the partially embedded graph.
\cref{prop:colored-faces} is used to show the maintenance of E\ref{prop:fixed-rotations} throughout the algorithm.
Let $\Emb_{\mathcal H}(G_i^+)$ be the set of all embeddings of $G_i^+$ that satisfy \cref{prop:half-edges,prop:fixed-rotations,prop:colored-faces},
and let $\Ord_{\mathcal H}(G_i^+)=\{\Ord(\mathcal E)\mid \mathcal E \in \Emb_{\mathcal H}(G_i^+)\}$ contain all orders of half-edges on the outer face from these embeddings.

\begin{figure}[t]
  \centering
  \begin{subfigure}{\linewidth}
    \includegraphics[page=1]{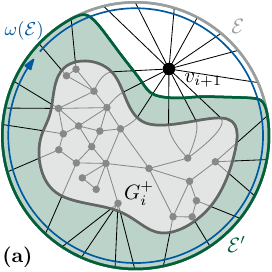}\hspace*{.1cm}\includegraphics[page=2]{fig/embeddings}\hspace*{.1cm}\includegraphics[page=6]{fig/embeddings}\phantomsubcaption\label{fig:bic-emb-graph}
    \phantomsubcaption\label{fig:bic-emb-tree}
    \phantomsubcaption\label{fig:bic-emb-update}
    \phantomsubcaption\label{fig:bic-emb-next}
  \end{subfigure}
  \vspace{-.8cm}
  \caption{
    \textbf{(a)}~A~drawing $\mathcal E'$ of $G_i^+$ with its half-edges in the green area. Adding the next vertex~$v_{i+1}$ yields drawing $\mathcal E$ of $G_{i+1}^+$.
    \textbf{(b)}~PC-tree $T_i$ representing all planar embeddings of $G_i^+$. The edges $F$ need to be made consecutive for $v_{i+1}$.
      Orders relevant for the proof of \Cref{lem:peplan-bicon} are shown in blue.
    \textbf{(c)}~PC-tree $\update{T_i}{F}$. Splitting this tree also splits C-node $\mu$ into $\mu'$ in $\usplit{T_i}{F}$ and $\mu^F$ in $\usplit{T_i}{F^c}$.
      The edges incident to $\mu^F$ are those in $E(T_i,F)$.
      If the bold edges of $v_{i+1}$ are considered fixed, they fix the flip of $\mu$ and thus also of $\mu^F$ and $\mu'$.
    \textbf{(d)}~PC-tree $T_{i+1}$ obtained by merging trees $T'$ and $S'$, where $\mu'$ is what remains of the split C-node $\mu$.
  }
  \label{fig:bic-emb}
  \vspace{-.2cm}
\end{figure}

\enlargethispage{1em}
\subparagraph*{Testing Planarity.}
To test the planarity of a biconnected graph $G$ given an st-ordering $v_1,\dots,v_n$ of its vertices, we compute PC-trees $T_1,\ldots,T_{n-1}$ satisfying the invariant $\Ord(G_i^+) = \Ord(T_i)$ for all $i\in\{1,\ldots,n-1\}$.
Tree~$T_1$ consists of a single P-node with leaves $E(v_1)$.
Given PC-tree $T_i$, the next PC-tree $T_{i+1}$ is obtained as follows; see also \Cref{alg:hl-bicon-plan}.
Conceptually, we make the half-edges $F$ that lead from $G_i^+$ to $v_{i+1}$ consecutive in $T_i$ and replace them by a single edge leading to a new P-node with leaves $E(v_{i+1})\setminus F$.
Formally, we first turn~$v_{i+1}$ into a PC-tree $S$ consisting of a single P-node with leaves $E(v_{i+1})$.
We make the edges $F$ between $G_i^+$ and $v_{i+1}$ consecutive in both $T_i$ and $S$ using the \Update operation; see \Cref{fig:bic-emb-tree,fig:bic-emb-update}.
We split the resulting PC-tree $\update{T_i}{F}$ into trees $T^F=\usplit{T_i}{F}[F^c]$ and $T'=\usplit{T_i}{F}$, where $T^F$ describes the order of half-edges $F$ leading from~$G_i^+$ to~$v_{i+1}$ and $T'$ describes the order of the remaining half-edges of $G_i^+$.
Similarly, we split $\update{S}{F}$ into $S^F=\usplit{S}{F}[F^c]$ and~$S'=\usplit{S}{F}$, where $S^F$ describes the order of half-edges $F$ leading from $v_{i+1}$ to $G_i^+$ and $S'$ describes the order of the remaining half-edges of $v_{i+1}$.
Note that $L(S^F)=L(T^F)=F\cup\{\ell\}$ and $L(S')\cap L(T')=\{\ell\}$.
Furthermore, $L(S')\cup L(T')$ contains all half-embedded edges that are present after step $i+1$ plus $\ell$, that is $L(S')\cup L(T')=L(T_{i+1})\cup\{\ell\}$.
Finally, we merge trees $S'$ and $T'$ at $\ell$ to obtain $T_{i+1}$ with $\Ord(T_{i+1})=\Ord(\merge{T'}{S'})$; see \Cref{fig:bic-emb-next}.
If none of the steps fails due to an impossible update, this means that we found (implicit) planar embeddings for all considered subgraphs.
Finding a non-null PC-tree $T_{n-1}$ then suffices to show planarity, as the edges of $v_n$ are always consecutive in $T_{n-1}$ (which has no further leaves).
Conversely, if the process fails at any step, this is due to a subgraph having no planar embedding in which all edges to the next vertex are consecutive, meaning that the graph is non-planar \cite{bl-tft-76}.

\subparagraph*{Testing Partially Embedded Planarity.}
To test an instance $(G,H,\mathcal H)$ of partially embedded planarity we basically apply the same algorithm,
but whenever we create a P-node for a vertex $v_i$, we now create a color-constrained P-node that also reflects the color-constraints of $v_i$.
There are two important additional differences to the ordinary planarity algorithm; see also lines \ref{line:intersect} to \ref{line:fix-flip} of \Cref{alg:hl-bicon-plan}.
First, the orders of half-edges around $v_{i+1}$ are now not only constrained by $T^F$, but also by~$v_{i+1}$, that is by \pc $S^F$.
We need to check that $T^F$ and $S^F$ allow for at least one common order of~$F$, that is whether $\Ord(T^F)\cap\overline{\Ord(S^F)}\neq\emptyset$ as the order of edges in $F$ entering $v_{i+1}$ is the reversal of the order in which they leave $G_i^+$.
Second, if $X=E(\update{T_i}{F}, F)$ is a set of edges that are consecutive around a C-node $\mu$ of $\update{T_i}{F}$, the constraints of $v_{i+1}$ may fix the order of $X$ around $\mu$ and thus the flip of~$\mu$; see \Cref{fig:bic-emb-update}.
Note that splitting $T_i$ then splits $\mu$ into a C-node~$\mu^F$ in~$T^F$ and a C-node~$\mu'$ in~$T'$, both incident to the leaf $\ell$ introduced by the split.
Both~$\mu^F$ and~$\mu'$ need to be fixed if the order of $X$ around $\mu$ is fixed by $v_{i+1}$.
We detect and handle this as follows.
After finding one order in $\Ord(T^F)\cap\overline{\Ord(S^F)}$ in the intersection test, we check whether the intersection also contains a second order where $\mu^F$ is flipped the other way.
If this is not the case, $v_{i+1}$ fixes $\mu$ and we accordingly fix the flip of~$\mu'$ in the copy $T''$ of $T'$; otherwise we set $T''$ to be equal to $T'$.
Finally, we merge the trees $S'$ and $T''$ at $\ell$ as before to obtain $T_{i+1}$.
Again, we can apply this process until we obtain a \pc~$T_{n-1}$ or otherwise conclude that the instance is negative.
Note that we here also need to perform the last step of the algorithm to check that the constraints of~$v_n$ are respected.

This concludes the description of our algorithm for testing \peplan.
See \Cref{alg:hl-bicon-plan} for high-level pseudo-code highlighting our changes and \iftoggle{long}{\Cref{alg:peplan-bicon} in \Cref{sec:peplan-details} for the full implementation.}{the full version \cite{} for a detailed implementation.}
For correctness of our algorithm, it remains to show that the above steps actually yield PC-trees that satisfy \cref{prop:half-edges,prop:fixed-rotations,prop:colored-faces}.
\iftoggle{long}{
  We give a full proof in \Cref{sec:peplan-details} and only sketch the main points here.
}{
  We will only sketch the main points here, the full proof can be found in the full version \cite{}.
}

\begin{restatable}\restateref{lem:peplan-bicon}{lemma}{lemPeplanBicon}\label{lem:peplan-bicon}
For every step $i\in\{1,\ldots,n\}$ of the algorithm, $\Ord_\mathcal{H}(G_i^+) = \Ord(T_i)$ holds.
\end{restatable}
\begin{proofsketch}
  The proof works analogously to the proof for the ordinary planarity test showing $\Ord(G_i^+) = \Ord(T_i)$ by induction on the number of steps.
  In the proof we will explicitly note where additional arguments are needed for the partial setting, otherwise the argumentation is taken directly from the ordinary planarity setting considering $\Ord(G_i^+)$ instead of $\Ord_\mathcal{H}(G_i^+)$.
  \cref{prop:half-edges} as well as \cref{prop:fixed-rotations,prop:colored-faces} in the partial case trivially hold for the first step, showing $\Ord_\mathcal{H}(G_1^+) = \Ord(T_1)$.
  Assuming that $\Ord_\mathcal{H}(G_i^+) = \Ord(T_i)$ holds for step $i$, the statement for the next step $i+1$ can be shown by arguing both inclusions separately.

  To show $\Ord_\mathcal{H}(G_{i+1}^+) \subseteq \Ord(T_{i+1})$, take an order $\sigma\in\Ord_\mathcal{H}(G_{i+1}^+)$ and let $\mathcal E\in\Emb_\mathcal{H}(G_{i+1}^+)$ be a corresponding embedding with $\Ord(\mathcal E)=\sigma$.
  Let $\mathcal E'$ be the embedding of $G_i^+$ obtained by deleting~$v_{i+1}$; see \Cref{fig:bic-emb-graph}.
  It can easily be shown that $\mathcal E'\in\Emb_\mathcal{H}(G_i^+)$ and we have, thanks to $G_i$ being connected, $\tau_1=\Ord(\mathcal E')\in\Ord_\mathcal{H}(G_i^+)$ and, by the inductive hypothesis, $\tau_1\in\Ord(T_i)$.
  All edges in $F$ must be consecutive in $\tau_1$ and we thus have $\tau_1 \in \Ord(\update{T_i}{F})$.
  For the ordinary planarity test, one can directly show $\sigma=\tau_1[F\to\ell]\otimes_\ell\tau_2[F\to\ell] \in \Ord(\merge{\usplit{T_i}{F}}{\usplit{S}{F}})=\Ord(T_{i+1})$ using the rotation $\tau_2$ of $v_{i+1}$ in $\mathcal E$; see~\Cref{fig:bic-emb-tree}.
  In the partial setting, $T_{i+1}$ may have instead been obtained from $T_i$ by fixing the flip of a split C-node~$\mu$, turning the last equality in this chain into a non-strict superset inclusion $\supseteq$.
  Fortunately, as $\tau_1$ and $\tau_2$ and thus also $\sigma$ are derived from a drawing, the flips $\tau_1$ and $\tau_2$ induce for the split parts of $\mu$ always line up with their fixed flip.
  Thereby, we can show that $\sigma$ is always contained in $\Ord(T_{i+1})$ for any $\sigma=\Ord(\mathcal E)\in\Ord_\mathcal{H}(G_{i+1}^+)$, concluding the proof for this direction.

  To conversely show $\Ord_\mathcal{H}(G_{i+1}^+) \supseteq \Ord(T_{i+1})$, take an order $\sigma \in \Ord(T_{i+1})$ and let $\sigma_1\in T'$ and $\sigma_2\in S'$ be such that $\sigma=\sigma_1\otimes_\ell\sigma_2$. To be able to apply the inductive hypothesis, we seek an order $\sigma_F\in \Ord(T^F)\,\cap\,\overline{\Ord(S^F)}$ of $F\cup\{\ell\}$
  such that $\tau_1=\sigma_1\otimes_\ell\sigma_F$ not only lies (by construction) in $\Ord(\merge {T'} {T^F})$, but also
  in $\Ord(\update{T_i}{F})\subseteq\Ord(T_i)$.
  For the ordinary planarity test this is easy, as $S$ and $S^F$ each consist of a single P-node, allowing arbitrary orders.
  In the partial setting, $S$ consists of a color-constrained P-node that may only allow certain orders, leading to the intersection $\Ord(T^F)\,\cap\,\overline{\Ord(S^F)}$ potentially being empty or not containing an order that is compatible with $\sigma_1$ to form an admissible order of $T_i$.
  In the algorithm, we explicitly guard against the intersection being empty.
  For the second issue recall that, by \Cref{lem:merge-splits}, an incompatiblity can only arise from the flip of $\mu^F$ in $T^F$ induced by $\sigma_F$ not coinciding with the flip of $\mu'$ induced by $\sigma_1$ in $T'$.
  In the algorithm, the intersection test either found orders for both flips of $\mu^F$ or otherwise the flip of $\mu'$ in $T'$ was fixed to the one of $\mu^F$.
  This ensures that we can always find a compatible $\sigma_F$ yielding a $\tau_1\in\Ord(T_i)$.

  By the inductive hypothesis, we have $\tau_1\in\Ord_\mathcal{H}(G_i^+)$ and there exists an embedding $\mathcal E'\in\Emb(G_i^+)$ with $\Ord(\mathcal E')=\tau_1$.
  Furthermore, we have $\tau_2=\sigma_2\otimes_\ell\overline{\sigma_F}\in\Ord(\update{S}{F})\subseteq\Ord(S)$.
  We choose $\tau_2$ as rotation for $v_{i+1}$ when adding it to $\mathcal E'$ to obtain a planar embedding $\mathcal E$ of $G_{i+1}^+$.
  We can show that this results in $\Ord(\mathcal E) =
  \tau_1[F\to\ell] \otimes_\ell \tau_2[F\to\ell] =
  \sigma_1 \otimes_\ell \sigma_2=\sigma\in \Ord(G_{i+1}^+)$, satisfying \cref{prop:half-edges}.
  This already concludes the proof in case of the ordinary planarity test, while for the partial setting we also have to show maintenance of \cref{prop:fixed-rotations,prop:colored-faces}.

  If $v_{i+1}\notin H$ this can easily be done, as $H_{i+1}^+$ is unchanged from $H_{i}^+$ for \cref{prop:fixed-rotations}.
  Further note that in his case, $v_{i+1}$ and all its incident edges, especially those in $F$, lie in the same face of $\mathcal H$.
  \Cref{prop:colored-faces} holding for all edges $F$ in $\mathcal E'$ ensures this face coincides with the restriction the edges of $v_{i+1}$ have, showing this property is maintained in $\mathcal E$.
  
  If $v_{i+1}\in H$, a case-by-case analysis of the edges incident to $v_{i+1}$ based on their types (i.e., fixed, restricted, and/or unrestricted) can be used to show that \cref{prop:colored-faces} is maintained.
  Regarding \cref{prop:fixed-rotations}, note that the constructed rotation for $v_{i+1}$ respects the constraints of $\mathcal H$.
  The relative position of $v_{i+1}$ with regard to $G_i$ can only be wrong if $v_{i+1}$ is not connected to any vertex of $H_i$, that is if $F$ contains no edges of~$H$.
  In this case, all edges in $F$ are restricted to be embedded in the same face, which is incident to $v_{i+1}$ and some vertices from $V_i$.
  Hence, in this case \cref{prop:colored-faces} ensures the correct relative positions and thus \cref{prop:fixed-rotations}.
  As all three properties are satisfied, we get $\sigma\in\Ord_\mathcal{H}(G_{i+1}^+)$ for both $v_{i+1}\notin V(H)$ and $v_{i+1}\in V(H)$, concluding our~proof.
\end{proofsketch}
\begin{prooflater}{proofPeplanBicon}
We prove this by induction on the number of steps.
For step $i=1$, observe that $T_1$ by construction allows the same rotations as $v_1$.
Thus, $\Ord_\mathcal{H}(G_1^+) = \Ord(T_1)$ holds.
For the inductive step, assume that $\Ord_\mathcal{H}(G_i^+) = \Ord(T_i)$ holds for step $i$.
We will show the statement for the next step $i+1$ by arguing both inclusions separately.

\subparagraph*{Direction $\Ord_\mathcal{H}(G_{i+1}^+) \boldsymbol{\subseteq} \Ord(T_{i+1})$.} 
Let $\sigma\in\Ord_\mathcal{H}(G_{i+1}^+)$ and let $\mathcal E\in\Emb_\mathcal{H}(G_{i+1}^+)$ be a corresponding embedding with $\Ord(\mathcal E)=\sigma$.
Let $\mathcal E'$ be the embedding of $G_i^+$ obtained by deleting $v_{i+1}$ together with its incident half-edges from $\mathcal E$, turning incident ordinary edges to half-edges; see \Cref{fig:bic-emb-graph}.
As~$\mathcal E\in\Emb_\mathcal{H}(G_{i+1}^+)$, it satisfies \cref{prop:half-edges,prop:fixed-rotations,prop:colored-faces}.
Note that due to \cref{prop:half-edges}, $v_{i+1}$ must be on the outer face if it has half-edges.
If it has none, we have $i+1=n$, $G_{i+1}^+$ contains no half-edges, and we can thus choose an arbitrary face incident to $v_n$ to be the outer one.
Removing $v_{i+1}$ and turning its incident edges into half-edges thus leaves all half-edges on the same face (the outer one) and $\mathcal E'$ thus satisfies \cref{prop:half-edges}.
As~$\mathcal E$ satisfies \cref{prop:fixed-rotations}, its restriction $\mathcal E'$ also does so.
To argue \cref{prop:colored-faces} we now consider a drawing $\mathcal E^\oplus$ of $\mathcal E$ on a disk, defined analogously to disk-drawing $\mathcal H_i^\oplus$ of $\mathcal H_i$.
We consider three different types of faces of $\mathcal H$ that are present in $\mathcal E^\oplus$.
Faces that are not incident to $v_{i+1}$, together with all edges of $G-H$ they contain, remain unchanged in $\mathcal E'^\oplus$, thus these edges still satisfy \cref{prop:colored-faces}.
Faces that are incident to $v_{i+1}$ but no vertex from $V_i$ are incident to the border of $\mathcal E^\oplus$ and only contain half-edges with $v_{i+1}$ as endpoint.
These faces together with all their contained edges are removed in $\mathcal E'^\oplus$ and they can thus not violate \cref{prop:colored-faces}.
Lastly, consider the set $\mathcal F$ of faces that are incident to $v_{i+1}$ as well as a vertex from $V_i$.
At most two of these may also be incident to the border of $\mathcal E^\oplus$, while the remaining ones are closed by~$v_{i+1}$ in $\mathcal E^\oplus$.
Note that in case $v_{i+1}\notin V(H)$, we have $|\mathcal F|=1$ as $v_{i+1}$ and its incident edges must lie entirely within one face of~$\mathcal H$.
In either case, all faces of $\mathcal F$ are also present in $\mathcal E'^\oplus$.
The edges between $v_{i+1}$ and $V_i$ turn into half-edges, the half-edges incident to $v_{i+1}$ are removed, while the half-edges incident to $V_i$ are retained.
As the assignment of these (half-)edges to faces remains unchanged, \cref{prop:colored-faces} is satisfied also in this last case.

As all three properties are satisfied in $\mathcal E'$, we thus have $\mathcal E'\in\Emb_\mathcal{H}(G_i^+)$.
As $G_i$ is connected by \cref{asm:connected} (the st-ordering ensures this), we can define $\tau_1=\Ord(\mathcal E')$.
Similarly, let $\tau_2$ be the order of all edges incident to $v_{i+1}$ in $\mathcal E$; see \Cref{fig:bic-emb-tree}.
As all half-edges are on the outer face of $\mathcal E$ and $\mathcal E'$, $F$ is consecutive both in $\tau_2$ and in $\tau_1$.
Observe that $\tau_1[F]=\overline{\tau_2[F]}$.
Since $\mathcal E$ can be obtained by combining $\mathcal E'$ with $v_{i+1}$ using~$\tau_2$ as rotation, order $\sigma$ can be obtained by merging $\tau_1$ and $\tau_2$ at $F$, that is $\sigma=\sigma_1\otimes_\ell\sigma_2$ for $\sigma_1=\tau_1[F\to\ell]$ and $\sigma_2=\tau_2[F\to\ell]$. 

As $\mathcal E'\in\Emb_\mathcal{H}(G_i^+)$ we have $\tau_1\in\Ord_\mathcal{H}(G_i^+)$ and, by the inductive hypothesis, $\tau_1\in\Ord(T_i)$.
All edges in $F$ are consecutive in $\tau_1$ and we thus have $\tau_1 \in \Ord(\update{T_i}{F})$.
As $\sigma_1=\tau_1[F\to\ell]$, it follows that $\sigma_1 \in \Ord(T')$ with $T'=\usplit{T_i}{F}$.
Note that $\tau_2\in\Ord(S)$ by construction of $S$.
As above, all edges in $F$ are consecutive in $\tau_2$ and thus $\tau_2 \in \Ord(\update{S}{F})$.
As $\sigma_2=\tau_2[F\to\ell]$, it follows that $\sigma_2 \in \Ord(S')$ with $S'=\usplit{S}{F}$.
Recall that $T_{i+1}=\merge{T''}{S'}$ where either $T''=T'$ or $T''$ is obtained from $T'$ by fixing the flip of C-node $\mu'$ adjacent to $\ell$.
If $T'=T''$ we directly have $\sigma=\sigma_1\otimes_\ell\sigma_2 \in \Ord(\merge{T'}{S'})=\Ord(\merge{T''}{S'})=\Ord(T_{i+1})$ as claimed due to $\sigma_1 \in \Ord(T')$ and $\sigma_2 \in \Ord(S')$.
Otherwise, we have $\Ord(\merge{T'}{S'})\supseteq\Ord(\merge{T''}{S'})$ and to show $\sigma_1\otimes_\ell\sigma_2 \in \Ord(\merge{T''}{S'})$ it suffices to show $\sigma_1 \in \Ord(T'')$.
That is, the flip that $\sigma_1$ induces on $\mu'$ coincides with the flip of $\mu$ dictated by $v_{i+1}$; see \Cref{fig:bic-emb-update}.
Note that the former is the same as the flip of $\mu$ induced by $\tau_1$, 
while the latter is the same as the flip of $\mu^F$ induced by $\overline{\tau_2[F^c\to\ell]}$.
As $\tau_1[F^c\to\ell]=\overline{\tau_2[F^c\to\ell]}$ and the projection of $\tau_1$ does not change the flip of $\mu$ it induces, both flips have to be the same and $\sigma_1\otimes_\ell\sigma_2 \in \Ord(\merge{T''}{S'})$.
This shows that $\sigma\in\Ord(T_{i+1})$ for any $\sigma=\Ord(\mathcal E)\in\Ord_\mathcal{H}(G_{i+1}^+)$ and thereby concludes the proof for this direction.

\subparagraph*{Direction $\Ord_\mathcal{H}(G_{i+1}^+) \boldsymbol{\supseteq} \Ord(T_{i+1})$.}
Let $\sigma \in \Ord(T_{i+1})$ and recall that, with $T'=\usplit{T_i}{F}$ and $S'=\usplit{S}{F}$, we have $\Ord(T_{i+1})\subseteq\Ord(\merge{T'}{S'})$, where equality holds if we did not fix the flip of~$\mu'$.
Note that here we used \cref{asm:consecutive} as we assume $\Ord(T'),\Ord(S')\neq\emptyset$.
Let $\sigma_1\in T'$ and $\sigma_2\in S'$ be such that $\sigma=\sigma_1\otimes_\ell\sigma_2$ following \Cref{lem:merge-splits}. With $T^F=\usplit{T_i}{F^c}$ and $S^F=\usplit{S}{F^c}$, let $\sigma_F\in \Ord(T^F)\,\cap\,\overline{\Ord(S^F)}$ be an order of $F\cup\{\ell\}$ where the induced flip of $\mu^F$ in $T^F$ coincides with the flip of $\mu'$ induced by $\sigma_1$ in $T'$.
Such an order has to exist as the intersection test either found orders for both flips of $\mu^F$ or otherwise the flip of $\mu'$ was fixed to the one of $\mu^F$.
We set $\tau_1=\sigma_1\otimes_\ell\sigma_F$ and $\tau_2=\sigma_2\otimes_\ell\overline{\sigma_F}$. If~$E(T_i,F)$ is a single edge, we directly get $\tau_1\in\Ord(\merge {T'} {T^F})=\Ord(\update{T_i}{F})\subseteq\Ord(T_i)$ due to \Cref{lem:merge-splits} as $\sigma_1\in\Ord(T')$ and $\sigma_F\in\Ord(T^F)$.
Otherwise, $\tau_1\in\Ord(\update{T_i}{F})$ only holds if the induced flips of $\mu'$ and~$\mu^F$ correspond to the same flip of $\mu$.
As we chose $\sigma_F$ to satisfy this, we also get $\tau_1\in\Ord(T_i)$ in this case.
Furthermore, we always have $\tau_2\in\Ord(\update{S}{F})\subseteq\Ord(S)$ as $E(\update{S}{F},F)$ is a single edge.
This is because it follows from \cref{asm:p-node} that the leaves~$F$ are all adjacent to the same P-node in~$S$ (which is in this case the only inner node of $S$).

By the inductive hypothesis, we have $\tau_1\in\Ord_\mathcal{H}(G_i^+)$ and there exists an embedding $\mathcal E'\in\Emb_\mathcal{H}(G_i^+)$ with $\Ord(\mathcal E')=\tau_1$.
We choose $\tau_2$ as rotation for $v_{i+1}$ and add it to $\mathcal E'$ to obtain an embedding $\mathcal E$ of $G_{i+1}^+$.
We thereby effectively complete the half-edges $F$ in~$\mathcal E'$ by connecting them to $v_{i+1}$ and insert the remaining edges of $v_{i+1}$ as new half-edges.
Regarding the orders of these edges, recall that $F$ is consecutive but oppositely ordered in $\tau_1$ and $\tau_2$.
This ensures that $\mathcal E$ is planar and has all half-edges on the outer face, that is \cref{prop:half-edges} is satisfied.
Since $(\sigma_1\otimes_\ell\sigma_F)[F\to\ell]=\sigma_1$ and similarly $(\sigma_2\otimes_\ell\sigma_F)[F\to\ell]=\sigma_2$, we thus get
\begin{align*}
  \Ord(\mathcal E) &=
  \tau_1[F\to\ell] \otimes_\ell \tau_2[F\to\ell] \\&=
  ((\sigma_1\otimes_\ell\sigma_F)[F\to\ell]) \otimes_\ell
  ((\sigma_2\otimes_\ell\overline{\sigma_F})[F\to\ell]) \\&=
  \sigma_1 \otimes_\ell \sigma_2=\sigma,
  \end{align*}
as order of half-edges on the outer face.
To show $\sigma\in\Ord_\mathcal{H}(G_{i+1}^+)$, it remains to show \cref{prop:fixed-rotations,prop:colored-faces}.
Both are satisfied in $\mathcal E'$ by the inductive hypothesis and in the $\sigma_2$-induced embedding of~$v_{i+1}$ by construction.
For their combination, we distinguish two cases depending on whether~$v_{i+1}$ is part of $H$ or not.
If $v_{i+1}\notin V(H)$, it must lie entirely within one face of $\mathcal H$.
Note that in this case, also all edges incident to $v_{i+1}$ are not in $H$ and thus must lie within this same face of~$\mathcal H$.
In particular, this holds for the edges in $F$.
If $f(e)=\bot$ for one $e\in F$, this holds for all edges incident to $v_{i+1}$ and \cref{prop:colored-faces} cannot be violated by any of the added edges.
Otherwise, \cref{prop:colored-faces} holding for $\mathcal E'$ already ensures that $e$ is embedded in face~$f(e)$.
As $v_{i+1}$ lies in the interior of the face $f(e)$, all its remaining edges are thus also embedded in the same, correct face, and \cref{prop:colored-faces} is satisfied.
As $v_{i+1}\notin V(H)$, adding it does not affect the restriction to~$H$ considered by \cref{prop:fixed-rotations}, which is thus also left satisfied.
Thus, all three properties are satisfied if $v_{i+1}\notin V(H)$.

Now consider the case $v_{i+1}\in V(H)$.
Note that all edges of $G-V(H)$ present in $\mathcal E'^\oplus$ still lie in the same face, leaving \cref{prop:colored-faces} unchanged.
Consider the newly-inserted half-edges of $G-V(H)$ that lie in a newly-created face incident to~$v_{i+1}$ as well as the border of $\mathcal E^\oplus$, but not to any vertex from~$V_i$.
For these edges, the order chosen by $\tau_2\in\Ord(S)$ ensures that \cref{prop:colored-faces} is satisfied.  
The remaining newly-inserted half-edges lie in one of the at most two faces incident to~$v_{i+1}$, (some vertices of)~$V_i$ and (two distinct segments of) the border of $\mathcal E^\oplus$, which we call \emph{boundary faces}.
Here, we distinguish whether $F$ contains an edge that is also in $H$.
If this is not the case, all edges of~$F$ lie in the same face of $\mathcal H$, which is also the single boundary face.
Note that for any edge, the faces incident to the left and right of its one end need to be the same as the faces incident to the left and right, respectively, of its other end.
This ensures that both $G_i^+ $ and $v_{i+1}$ agree on the face in which~$F$ should be embedded and \cref{prop:colored-faces} is satisfied.
If $F$ contains at least one edge that is also part of $H$, inserting $v_{i+1}$ may close some faces of $\mathcal H$.
Note that all edges contained in these faces satisfy \cref{prop:colored-faces} in $\mathcal E'^\oplus$ and also do so in $\mathcal E^\oplus$, where their incident segment of the border of $\mathcal E'^\oplus$ was effectively contracted into a single point.
These faces may contain no newly-inserted half-edges, and all old half-edges are completed to $v_{i+1}$.
In contrast to this, the up to two boundary faces may contain half-edges completed by $v_{i+1}$ as well as old half-edges of $G_i^+$ that were not yet completed and newly-inserted half-edges originating from $v_{i+1}$.
The boundary face is also incident to at least one edge that is both in $F$ and in $H$, which ensures that $v_{i+1}$ and~$G_i^+$ agree on the face in which to embed all these edges, satisfying \cref{prop:colored-faces}.

Regarding \cref{prop:fixed-rotations}, note that the construction of the rotation for $v_{i+1}$ ensures that the constraints of $\mathcal H$ are respected for this newly-inserted vertex.
The relative position of $v_{i+1}\in V(H)$ with regard to $G_i$ is only relevant if there are vertices $V_i\cap V(H)\neq\emptyset$ that are not connected to $v_{i+1}$ in $\mathcal H_i$.
As \cref{prop:fixed-rotations} holds for the embedding $\mathcal E'$ of $G_i$ and thereby also for all connected components of $H_i$ in it, it suffices to consider the case where $v_{i+1}$ is a new connected component, i.e., $F$ contains no edges of~$H$.
In this case, all edges in $F$ are restricted to be embedded in the same face, which is incident to $v_{i+1}$ and some vertices from $V_i$.
Hence, in this case \cref{prop:colored-faces} ensures the correct relative positions and thus \cref{prop:fixed-rotations}.
Otherwise, that is if $G_i$ contains no vertices of $H$ or if $F$ contains a fixed edge, the relative position of $v_{i+1}$ cannot violate \cref{prop:fixed-rotations}.
As all three properties are satisfied, we get $\sigma\in\Ord_\mathcal{H}(G_{i+1}^+)$ for both $v_{i+1}\notin V(H)$ and $v_{i+1}\in V(H)$.
This concludes the proof of $\Ord_\mathcal{H}(G_{i+1}^+) \supseteq \Ord(T_{i+1})$.
\end{prooflater}

Interestingly, the interactions between the two halves of a split C-node, which we need to explicitly handle by fixing the one half if the other one is implicitly fixed, are also one of the main concerns of the algorithm by Chiba et al.\ \cite{cnao-ala-85} for generating an embedding in addition to testing planarity.
Hence, after successfully running our \peplan testing algorithm, an embedding can be generated using the same approach.

\subsection{Non-Biconnected Instances}\label{sec:peplan-nonbic}
The ordinary planarity test by Booth and Lueker can be applied to non-biconnected graphs by simply processing each biconnected component independently.
This approach unfortunately cannot directly be applied for \peplan, as we also need to account for the constraints of cut-vertices.
Instead of relying on an involved preprocessing step, we extend our testing algorithm to directly handle non-biconnected inputs using the generalized planarity test by Haeupler and Tarjan \cite{ht-pav-08}, which we describe in the following.

When applying the planarity test to a non-biconnected instance, we can no longer assume $v_1,\dots,v_n$ to be an st-ordering that ensures that both $G[V_i]$ and $G[V\setminus V_i]$ are connected for every $i=1,\ldots,n$.
Haeupler and Tarjan retain the property that at least $G[V\setminus V_i]$ is connected by using a leaf-to-root ordering of a DFS-tree \cite{ht-pav-08}.
Thus, one can still assume all half-embedded edges to lie on the outer face.
But, at every step of the algorithm, we may now have multiple distinct connected components in~$G_i^+$, each represented by their own PC-tree.
When inserting a next vertex $v_{i+1}$, this may now cause previously distinct connected components to merge.
Note that this may happen independently of whether $v_{i+1}$ is a cut-vertex in $G$ whenever $v_{i+1}$ separates multiple subtrees of the DFS-tree.
The generalized algorithm handles this case by incrementally merging the components $C_1,\ldots,C_k$ of $G_i^+$ that are adjacent to~$v_{i+1}$ into the tree~$S_0$ representing~$v_{i+1}$.
To combine the \pc \TCj[j+1] of the next component $C_{j+1}$ with the current tree $\Sj$ into the next tree $\Sj[j+1]$,
we can use the same process as we used for combining $T_i$ with $S$ into $T_{i+1}$ in the biconnected setting.
The final tree $\Sj[k]$ then represents the component of $v_{i+1}$ in $G_{i+1}^+$.

There is a second issue that needs consideration in the non-biconnected partial setting.
Even if the instance is positive, an update may now fail if the constraints of $v_{i+1}$ require us to nest some incident \emph{blocks} (i.e., maximal biconnected components) and we process an outer block before the nested one.
Such nesting may be enforced by the fixed order of edges of $\mathcal H$ or by the color-constraints around $v_{i+1}$\iftoggle{long}{; see \Cref{fig:pe-blocks}}{}.
Fortunately, if a block needs to be nested within another block, it may have no further half-embedded edges for the instance to be positive.
Thus, processing any nested block before the block is nested within, that is using an inside-out nesting order, ensures any nested blocks are processed and thereby entirely removed first and the edges of an outer blocks can afterwards be made consecutive.

\iftoggle{long}{\Cref{lem:peplan-nonbic} in \Cref{sec:peplan-nonbic-details} shows}{In the full version \cite{}, we show} that incorporating these two changes breaks no assumptions we made in the biconnected setting while being sufficient to handle non-biconnected instances.
There, we also present pseudo-code for the full \peplan test for not-necessarily biconnected instances.
Altogether, this yields the following theorem.
\begin{restatable}\restateref{thm:main}{theorem}{thmMain}
  \label{thm:main}
  An instance $(G,H,\mathcal H)$ of \peplan can be tested in time linear in the size of $G$.
\end{restatable}

\section{Summary}\label{sec:summary}
\enlargethispage{1em}
This concludes the exposition of the core points of our algorithm for testing \peplan.
The remainder of this work presents the technical details that are required for the different parts of our algorithm, including full proofs of correctness.
\Cref{sec:pctree-details} gives an in-depth background on the operations of the base PC-tree data structure and describes our modifications to the \Update in the color-constrained case in detail.
In \Cref{sec:peplan-details}, we present the full proof of correctness for biconnected \peplan instances.
In \Cref{sec:peplan-nonbic-details}, we elaborate on our extension to not-necessarily biconnected instances. Finally, we describe the technical details required for a linear-time implementation in \Cref{sec:linear-time}.
Concluding remarks can be found in \Cref{sec:conclusion}.\iftoggle{long}{}{\todo{merge with summary}}
\toggletrue{myapx}

\iftoggle{long}{
\section{Full Details on PC-Tree Operations}\label{sec:pctree-details}
\pcOps

The following lemma describes the situations in which a PC-tree \Merge is the converse operation of a previous \Split and how this affects orders derived from the trees.
Recall that orders $\sigma_1\in \Ord(T')$ and $\sigma_2\in \Ord(T'')$ are compatible if $\merge{\sigma_1}{\sigma_2}\in\Ord(T)$ with $T'$ and $T''$ obtained from splitting $T$.
\lemMergeSplits*\label{lem:merge-splits*}
\proofMergeSplits

\pcUpdate

\section{Partially Embedded Planarity}\label{sec:peplan-details}
In this section, we want to give full and detailed proofs of correctness for our \peplan algorithms from \Cref{sec:peplan}.
First, we will again only consider the biconnected setting.
The generalization to not-necessarily biconnected graphs is considered in \Cref{sec:peplan-nonbic-details}.
\Cref{alg:peplan-bicon} contains the full pseudo-code for our solution to \peplan on biconnected graphs as described in \Cref{sec:peplan-bicon}.
See \Cref{alg:hl-bicon-plan} for a high-level comparison with the basic planarity test.

In our algorithm and especially its following proof of correctness, we want to make explicit the usages of three assumptions regarding the trees generated by the algorithm.
\begin{enumerate}[(T1)]
\item If the instance is positive, for each step $i$, the leaves $F$ can be made consecutive in $T_i$ as well as in~$S$. \label[assumption]{asm:consecutive}
\item The graph $G_i^+$ that $T_i$ represents is connected. \label[assumption]{asm:connected}
\item The leaves $F$ are all adjacent to the same P-node in $S$. \label[assumption]{asm:p-node}
\end{enumerate}
Observe that these three assumptions are trivially satisfied in the biconnected case we currently investigate.
For the non-biconnected case we consider in the next section, we will need to take more care to show that these assumption still hold.
Furthermore note that we have $L(T_n)\setminus F = \emptyset$ in the last step of our algorithm, a situation which will also appear more often throughout the algorithm in the non-biconnected case.
This no problem though, as we will never assume $F^c\neq\emptyset$ in our algorithm or its proof of correctness.

\begin{algorithm}[t]
$v_1, \ldots, v_n\leftarrow\texttt{st-Order}(G)$\;
$\mathcal T_1\leftarrow \texttt{vertexToPNode}(v_1)$\tcp*[l]{single P-node copying constraints~of~$v_1$}
\For{$i$ in $1, \ldots, n-1$}{
  $S\leftarrow \texttt{vertexToPNode}(v_{i+1})$\;$F\leftarrow$ edges between $G_i$ and $v_{i+1}$ in $G$\;

  $\molap{S'}{S^F} \leftarrow\usplit{S}{F}[\molap{F}{F^c}]$\;$           S^F  \leftarrow\usplit{S}{F}[          F^c ]$\;

  $\molap{T'}{T^F} \leftarrow\usplit{T_i}{F}[\molap{F}{F^c}]$\;$           T^F  \leftarrow\usplit{T_i}{F}[          F^c ]$\;

  \If{$\Ord(T^F)\cap\overline{\Ord(S^F)} = \emptyset$}{\Return \texttt{false}\;}

  \eIf{
    $S^F$ fixes C-node $\mu^F$ of $T^F$ in the intersection
  }{
    $T''\leftarrow T'$ with fixed respective C-node $\mu'$\;
  }{
    $T''\leftarrow T'$\;
  }
  $T_{i+1}\leftarrow \merge{S'}{T''}$\;
}
\Return \texttt{true} if no \Update returned the null tree, and \texttt{false} otherwise\;
\caption{
  Test a biconnected graph $G$ for \peplan.
}
\label{alg:peplan-bicon}
\end{algorithm}

\lemPeplanBicon*\label{lem:peplan-bicon*}
\proofPeplanBicon

\subsection{Non-Biconnected Instances}\label{sec:peplan-nonbic-details}
We now want to consider not-necessarily biconnected instances, for which \Cref{sec:peplan-nonbic} already sketched the changes that are neccesary.
In this section, we give an elaborate description of these changes together and show their correctness.
A detailed pseudo-code implementation of our overall algorithm is given in \Cref{alg:peplan-nonbic}.

\begin{algorithm}[tb]
  $v_1, \ldots, v_n\leftarrow\texttt{vertexOrder}(G)$\tcp*[l]{leaf-to-root DFS-tree order}
  $\mathcal T\leftarrow$ empty map from connected component to PC-tree\;
  $\mathcal T[\{v_1\}]\leftarrow \texttt{vertexToPNode}(v_1)$\tcp*[l]{single P-node copying constraints~of~$v_1$}
  \For{$i$ in $1, \ldots, n-1$}{
    $\Sj[0]\leftarrow\texttt{vertexToPNode}(v_{i+1})$\;
    $C_1, \ldots, C_k\leftarrow\texttt{blockOrder}(v_{i+1})$\label{line:blockOrder}\tcp*[l]{process nested blocks first}
    \For{$j$ in $1, \ldots, k$}{
      $F_j\leftarrow$ edges between $C_j$ and $v_{i+1}$\;
      $\molap{\Sj'}{\Sj^F} \leftarrow\usplit{\Sj[j-1]}{F_j}[\molap{F_j}{F_j^c}]$\;$             \Sj^F  \leftarrow\usplit{\Sj[j-1]}{F_j}[            F_j^c ]$\;
      \eIf{$F_j \neq L(\TCj)$}{
        $\molap{T'}{T^F} \leftarrow\usplit{\TCj}{F_j}[\molap{F_j}{F_j^c}]$\;$           T^F  \leftarrow\usplit{\TCj}{F_j}[            F_j^c ]$\;
        \If{$\Ord(T^F)\cap\overline{\Ord(\Sj^F)} = \emptyset$\label{line:intersect2}}{\Return \texttt{false}}
        \eIf{
          $\Sj^F$ fixes C-node $\mu^F$ of $T^F$ in the intersection
        }{
          $T''\leftarrow T'$ with fixed respective C-node $\mu'$\;
        }{
          $T''\leftarrow T'$\;
        }
        $\Sj\leftarrow \merge{\Sj'}{T''}$\;\label{line:merge}
      }{
        remove the common leaf $e_j$ from both $\Sj^F$ and $\Sj'$\;\label{line:rem-block}
        \lIf{$\Ord(\TCj)\cap\overline{\Ord(\Sj^F)} = \emptyset$}{\Return \texttt{false}}
        $\Sj\leftarrow \Sj'$\tcp*[l]{without $e_j$}
      }
    }
    $\mathcal T[\{v_i\}\cup C_1\cup\ldots\cup C_k]\leftarrow \Sj[k]$\;
  }
  \Return \texttt{true} if no \Update returned the null tree, and \texttt{false} otherwise\;

  \caption{
    Test a general (i.e., not-necessarily biconnected) graph $G$ for \peplan.
  }
  \label{alg:peplan-nonbic}
\end{algorithm}

Recall that, on not-necessarily biconnected instances, we can no longer use an st-ordering, but will instead use an leaf-to-root ordering of a DFS-tree as proposed by Haeupler and Tarjan \cite{ht-pav-08}.
This still ensures that $G[V\setminus V_i]$ is connected and all half-edges have to lie on the outer face, but we may now have multiple distinct connected components in~$G_i^+$, each represented by their own PC-tree.
When inserting a next vertex $v_{i+1}$, this may now cause previously distinct connected components to merge.
We handle this case by incrementally merging the components $C_1,\ldots,C_k$ of $G_i^+$ that are adjacent to~$v_{i+1}$ as follows.
We consider $v_{i+1}$ as initial component $C_0=\{v_{i+1}\}$ and observe that the union of all components yields the component $C=C_0\cup C_1\cup\ldots\cup C_k$ of $v_{i+1}$ in $G_{i+1}^+$.
We will compute \pcs $\Sj[1],\ldots,\Sj[k]$ satisfying the invariant $\Ord(\Sj) = \Ord_\mathcal{H}(G[C_0 \cup \cdots \cup C_j]^+)$ for $j=1,\ldots,k$; see the lines following line \ref{line:blockOrder} in \Cref{alg:peplan-nonbic}.
Note that we will thereby get $\Ord(\Sj[k]) = \Ord_\mathcal{H}(G[C_0 \cup \cdots \cup C_k]^+) = \Ord_\mathcal{H}(C)$ at the end of the iteration.
We obtain tree $\Sj[0]$ with $\Ord(\Sj[0]) = \Ord_\mathcal{H}(G[C_0]^+)$ similar to before by converting $v_{i+1}$ into a single P-node and copying its constraints.
We use $\mathcal T$ to map from components to their already computed \pcs, setting $\mathcal T[C]=\Sj[k]$ every time we have processed all $k$ components incident to a vertex~$v_{i+1}$.
To combine the \pc $\TCj$ of the next component $C_j$ with the tree $\Sj[j-1]$ from the previous iteration into the next tree \Sj,
we use the same process as we used for combining $T_i$ with $S$ into $T_{i+1}$ in our test for biconnected instances; see line \ref{line:merge} of \Cref{alg:peplan-nonbic}.
Note that while we now run the process on different trees, we will show in a moment that they still satisfy \Cref{asm:connected,asm:p-node}, i.e., that the component $C_j$ that $\TCj$ represents is connected and the leaves $F$ are all adjacent to the same P-node in $\Sj$.

\begin{figure}[t]
  \begin{subfigure}{.39\linewidth}
    \includegraphics[page=3,scale=1.3]{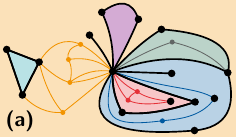}
  \end{subfigure}\begin{subfigure}{.37\linewidth}
    \centering
    \includegraphics[page=2,scale=1.3]{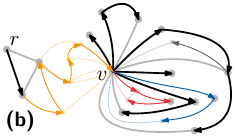}
  \end{subfigure}\begin{subfigure}{.24\linewidth}
    \hspace*{\fill}
    \includegraphics[page=4,scale=1.2]{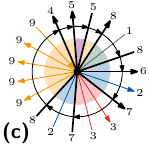}
  \end{subfigure}
  \caption{
    \textbf{(a)} An instance of \peplan with $H$-bridges colored according to the face they have to be embedded in.
    \textbf{(b)} A DFS tree on the underlying graph~$G$ with tree-edges directed away from the root $r$.
    \textbf{(c)} The \pc $S_0$ of vertex $v$, also indicating the directions of incident edges. The numbers indicate an order in which the incident blocks can be processed.
    Due to the restricted edges, blocks 3 and 2 need to be removed before blocks 7 and 8, respectively, can be made consecutive.
    Due to the fixed edges, block 7 needs to be processed before block 8.
Note that the currently shown rotation does not have block 2 consecutive.
  }
  \label{fig:pe-blocks}
\end{figure}

It is \Cref{asm:consecutive}, that is that the leaves $F$ can always be made consecutive in $\TCj$ and $\Sj$ if the instance is positive, which still needs consideration.
Even if the instance is positive, the update may now fail if the constraints of $v_{i+1}$ require us to nest some incident blocks and we process an outer block before the nested one; see \Cref{fig:pe-blocks}.
Fortunately, we can circumvent this issue as the nested blocks in positive instances need to have a certain structure.
Consider a cut-vertex $v_{i+1}$ with an incident block $C_j$.
When processing $C_j$, the component has no further half-edges except those leading to cut-vertex $v_{i+1}$, i.e., we have $F_j = L(\TCj)$.
We add no new half-edges to $\Sj[j-1]$ and remove all leaves $F_j$ without replacement after making them consecutive.
Thus, no part of $\TCj$ is present in $\Sj$ and we have $L(\Sj)\subsetneq L(\Sj[j-1])$; see line \ref{line:rem-block} in \Cref{alg:peplan-nonbic}.
Recall that we never assumed $F^c$ to be non-empty during the proof of \Cref{lem:peplan-bicon}, thus this does not affect correctness.

If a block now needs to be nested at $v_{i+1}$, it may not have further half-edges except for those leading to $v_{i+1}$ for the instance to be positive.
Thus, processing any nested block before the block is nested within, that is using an inside-out nesting order, ensures any nested blocks are always processed and removed first and the edges of their containing blocks can afterwards be made consecutive.
\Cref{lem:lt-blocks} in \Cref{sec:linear-time} shows how such an order can be found in time linear in the degree of $v_{i+1}$.

\begin{lemma}\label{lem:peplan-nonbic}
Let $C_1,\ldots,C_k$ be the connected components of $G_i^+$ that are incident to $v_{i+1}$, ordered according to their inside-out nesting enforced by the constraint of~$v_{i+1}$.
For every step $j\in\{0,\ldots,k\}$, $\Ord(T_i^j) = \Ord_\mathcal{H}(G[C_0 \cup \cdots \cup C_j]^+)$ holds.
\end{lemma}
\begin{proof}
The correctness of the statement can be shown analogously to the correctness of \Cref{lem:peplan-bicon}.
To be able to apply this proof, we still to show that its three underlying \Cref{asm:consecutive,asm:connected,asm:p-node} still hold for the new trees we use.
Processing the blocks in an inside-out nesting order ensures that, in a yes-instance, blocks that need to be nested are processed before the blocks they are nested within and this process removes all half-edges to the nested blocks.
This ensures that the half-edges to the outer blocks, which come later in the block order, can also be made consecutive, that is \Cref{asm:consecutive} is fulfilled.
As the component $C_{j+1}$ that $\TCj[j+1]$ represents is a block incident to $v_{i+1}$, it is also connected and thus $\TCj[j+1]$ satisfies~\Cref{asm:connected}.
The fact that the leaves $F$ are all adjacent to the same P-node in $\Sj$, that is that $\Sj$ satisfies \Cref{asm:p-node}, can be shown as follows.
The assumption holds per construction for \Sj[0].
For any later step $j$ with tree \Sj, note that the leaves $F$ we make consecutive were already present in \Sj[0], but they where never part of a set we made consecutive in an earlier step.
As the PC-tree update only modifies leaves that are made consecutive, the leaves in $F$ are thus all still incident to the same P-node they were incident to in \Sj[0].
\end{proof}

\section{Linear-Time Implementation}\label{sec:linear-time}
In this section, we show how the different parts of our algorithm for testing \peplan can be implemented to run in linear time.
We assume the usual representation of a graph using doubly-linked adjacency lists.
We further assume each vertex and edge has a label whether it is contained in $H$ and that the rotation system of $\mathcal H$ is given as separate doubly-linked adjacency lists.
Additionally, each edge~$e$ of~$H$ has pointers to objects representing the incident faces at both sides, which we will use as values of $f_v(e)$ and $f_u(e)$ for the endpoints $u,v$ of $e$.
Conversely, the face objects have, for each connected component incident to the face, a pointer to one of their incident edges in the component.
Note that we assume that this data structure represents a planar embedding, i.e., we have no cyclically nested faces and~components.

For our \pcs we assume a suitable implementation of the underlying PC-tree data structure, which especially allows merging C-nodes in constant time as required for an amortized-linear \Update \cite{hm-pta-03,hm-ptp-04,fpr-eco-21}.
For a linear-time implementation of \Update, the PC-trees need to be rooted.
As a consequence, \Merge can only be performed in constant time if the leaf~$\ell$ at which we merge is (incident to) the root of at least one of the two trees.
Fortunately, the planarity test of Haeupler and Tarjan, on which our algorithm is based, ensures this property \cite{ht-pav-08}.
Note that, for example, the graph data structures from the OGDF \cite{cgj-tog-13} together with the PC-tree implementation from \cite{fpr-eco-21} provide a C++ implementation suitable for our purposes.
An implementation of the Haeupler-Tarjan planarity test together with an embedder already exists based on these libraries \cite[Chapter~8.4]{fin-cpa-24}.

Similar to our graph representation, we use a second doubly-linked list to store fixed edges and their cyclic order in the \pcs.
To keep track of the colors of the angles following fixed edges and of the restricted edges around a node, we equip each fixed edge (representing its following angle) and restricted edge with a pointer to a shared counter for their node.
All objects of the same color at the same node have a pointer to the same counter, which separately counts angles and edges at this node.
Note that we do not maintain an index of the colors present at a node, but only an unordered list of all counters present at the node.
A counter of a certain color can thus only be accessed by linearly searching through the counter list, or in constant time via an object of the appropriate color.
This is no problem though, as this structure now easily allows decrementing the respective counter when removing an angle or edge from a node.
To create new nodes (e.g. when splitting), we keep a single global array with one entry per color that temporary allows looking up counters by color (and subsequently incrementing them appropriately), but only for one node at a time, i.e., the currently created one.
After the node is created, we reset the global array in time linear in the degree of the created node.
The counters now allow us to recognize the case when removing the last angle or edge of a color from a node, while not negatively affecting the asymptotic running time of the \Update operation.

\begin{lemma}\label{lem:lt-update}
  Method \Update on \pcs can be implemented to run in amortized time linear in $|A|$.
\end{lemma}
\begin{proof}
  As the base update procedure for ordinary PC-trees (see \Cref{alg:rcpc-update}) has an amortized running time linear in the number of full leaves \cite{hm-pta-03,hm-ptp-04,fpr-eco-21}, we want to show the same also holds for our modified version.
  Note that to meet the linear time bound, we can only spend a linear amount of time on the full neighbors for each full or partial node, while me may not process all their empty neighbors.

  Recall that we leave the labeling and terminal path finding in step \ref{pc-upd-labeling} unchanged.
  The consecutivity check of step \ref{pc-upd-reorder} (see also \Cref{alg:rcpc-flips}) can be implemented by keeping, for each P-node, a linked list of full fixed edges, and checking the predecessor and successor of every such edge after the labeling is complete.
  For each list, at most one edge may have a non-full predecessor and at most one a non-full successor for the full fixed edges to be consecutive.
  This also allows us to check that both of the at most two terminal fixed edges are directly adjacent to this block of full fixed edges, or the other partial fixed edge if the block is empty.
  Similarly, we can check that all P-nodes and fixed C-nodes have their fixed full edges on the same side of the terminal path.

  In step \ref{pc-upd-split}, we initialize and update the angle and restricted edge color counters appropriately during the splits.
  These counters then allow us to efficiently detect when a restricted edge got separated from all angles it could be embedded in or whether one of the split halves received no restricted edges.
  The second split that separates all empty edges is equivalent to splitting off the at most 2 partial neighbors together with the edge leading to the newly-created P-node with all full neighbors, and can thus be implemented in constant time without processing empty edges.
  After both splits, checking the admissible orders of $\mu^M$ and changing its type appropriately can be done in constant time as it has degree at most 4.
  The contractions in step \ref{pc-upd-contract} can then, thanks to the constant-time \Merge of C-nodes, be done in time linear in the length of the terminal~path.
  This shows that all our modifications do not increase the asymptotic running time of the \Update on \pcs.
\end{proof}

In addition to \Update, we also use a restricted form of the \Intersect method in our planarity test.
Fortunately, this test for a non-empty intersection can easily be implemented in linear time.

\begin{lemma}\label{lem:lt-intersect}
  The test whether $\Ord(T^F) \cap \overline{\Ord(S^F)} \ne \emptyset$ can be implemented to run in time linear in the number of leaves of $S^F$ and $T^F$.
\end{lemma}
\begin{proof}
  In our implementation, we use that both trees originate from an instance of \peplan, or, more precisely, that the trees that $T^F$ and $S^F$ stem from satisfy \cref{prop:colored-faces} and \cref{asm:p-node}, respectively.
  On the one hand, this means that $T^F$ already ensures that all restricted edges are embedded in the right angles, i.e., faces.
  On the other hand, $S^F$ consists of a single P-node.
  Furthermore, all fixed edges of $S^F$ are also fixed in $T^F$ and have the same incident faces.
  Similarly, the restricted edges are the same in both trees and they also have the same colors.
  Thus, the only way to have an empty intersection is if the rotation of the fixed edges of~$S^F$ is not admissible by $T^F$.
  We can test this by temporarily removing all non-fixed leaves from $T^F$ and checking whether the fixed order of $S^F$ is admissible by the resulting tree. 
  This can easily be checked in linear time, e.g. using an approach similar to the intersection of ordinary PC-trees.
\end{proof}

Recall that as second modification to the general planarity test, we need to check whether the intersection with the constraints of $v_{i+1}$ fixes the flip of a C-node $\mu$ of $T_i$ that is incident to the edges in the set $X=E(\update{T_i}{F}, F)$.
More precisely, we check whether the intersection with $S^F$ fixes the split-off half~$\mu^F$ of $\mu$ in $T^F$ and we therefore also need to fix the other half $\mu'$ in~$T'$.
This can be checked in linear time by performing the test from \Cref{lem:lt-intersect} twice, once fixing $\mu^F$ to its one flip and once to its other flip.
We report an empty intersection if both tests fail, fix the flip of $\mu'$ accordingly if only one of the two test runs succeeds, and leave $\mu'$ unmodified otherwise.
Both~$\mu^F$ and~$\mu'$ can easily be identified as they are incident to the leaf~$\ell$ introduced by the split.

The last building block we need for our linear-time algorithm is a procedure to find the block nesting order we use in \Cref{sec:peplan-nonbic}.
Recall that if the added vertex $v_{i+1}$ is a cut-vertex, we need to take special care about its required nesting of incident blocks, as we need to process nested blocks before the blocks they are nested in.
Note that~$v_{i+1}$ is a cut-vertex if and only if there is at least one component that has no further half-edges except those leading to $v_{i+1}$, and each such component corresponds to a block around $v_{i+1}$.
In this case, the components with remaining half-edges together with the half-edges of $v_{i+1}$ leading to later vertices form an additional block $B_r$, as they are all connected via the unembedded but connected graph $G[V\setminus V_i]$.
We will ensure that this block comes last in our generated order.

\begin{lemma}\label{lem:lt-blocks}
  In time linear in the degree of vertex $v_{i+1}$
  we can find an order $C_1, \ldots, C_k$ of the blocks incident to $v_{i+1}$
  such that, whenever the constraints of $v_{i+1}$ require a block $C_a$ to be nested within a block $C_b$, we have $a<b$.
  Furthermore, $C_k$ contains all components with remaining half-edges together with the half-edges of~$v_{i+1}$ leading to later vertices.
\end{lemma}
\begin{proof}
  We will put all blocks without fixed edges first (except for $B_r$), as these cannot force other blocks to be nested within them.
  Note that in a yes-instance, a block can only contain restricted edges with different colors if it also contains fixed edges with appropriate incident angles, as the edges could otherwise not be made consecutive.
  It remains to generate a subsequent order of blocks with fixed edges, where the prescribed order of fixed edges $\rho_{v_{i+1}}$ may force blocks to be nested.
  To do so, we will process the fixed edges in the order of $\rho_{v_{i+1}}$ and keep a stack of blocks for which we have seen some, but not all fixed edges.
  When encountering the last edge from a block, we remove the block and append it to the processing order of blocks.
  For example, this yields the block order shown in \Cref{fig:pe-blocks} when processing starts at the block with number 4.
  Note that in a yes-instance, two different blocks may not alternate and we can report a negative instance when we encounter a fixed edge of a block that is within, but not at the top of the stack.

  It remains to ensure that the block $B_r$ with the half-edges of $v_{i+1}$ can be put last in the order, which we do by appropriately choosing the edge from which we start the processing of fixed edges in their cyclic order $\rho_{v_{i+1}}$.
  If $B_r$ contains a fixed edge incident to $v_{i+1}$, we start processing the cyclic order $\rho_{v_{i+1}}$ with the fixed edge following thereafter.
  This already ensures that $B_r$ is the last block returned from the stack-based algorithm.
  If $B_r$ contains no fixed but a restricted edge~$e$ incident to $v_{i+1}$, we start processing from an arbitrary fixed edge following an angle with the same color as $e$ (which holds for block 4 and the topmost edge of block 8 in \Cref{fig:pe-blocks}).
  We then append $B_r$ and note that the choice of first fixed edge ensures that the we still have an appropriately-colored angle available when processing $B_r$ as last block.
  Lastly, if~$B_r$ has only unrestricted edges incident to $v_{i+1}$, we can start the processing of fixed edges at any point.
\end{proof}

Altogether, this now allows us to implement our full algorithm in linear~time.
\thmMain*\label{thm:main*}
\begin{proof}
  Our linear-time algorithm for testing \peplan on general instances works in the following steps.
  As preprocessing, we first apply two linear-time algorithms of Angelini et al.\ that allow us to process each connected component of $G$ separately and that compute the color-constraints $f(e)$ for each edge of an $H$-bridge \cite[Theorem 4.14 and Lemma 2.2]{abf-tpo-15}.
  We then run the algorithm described in \Cref{sec:peplan-nonbic-details} (whose correctness we established in \Cref{lem:peplan-nonbic}) on each connected component of~$G$ using the amortized linear-time \Update from \Cref{lem:lt-update}.
  Haeupler and Tarjan \cite[Section 4]{ht-pav-08} describe how the DFS can be implemented to at the same time yield the separate blocks incident to a vertex $v_{i+1}$ together with the edges connecting them to $v_{i+1}$.
  The sorting of the incident blocks in line \ref{line:blockOrder} of \Cref{alg:peplan-nonbic} is done using the approach from \Cref{lem:lt-blocks}.
  In line \ref{line:intersect2}, we use the intersection check from \Cref{lem:lt-intersect} and afterwards, if necessary, fix the respective C-node of the tree $T'$ representing the added block.
  Note that this is not needed if the current block has no further half-edges, as we can simply remove the respective subtree in this case.
  Further note that the DFS also ensures that the leaf $\ell$ at which we merge in line \ref{line:merge} is the root of $T'$, which allows it to be attached to $S'$ in constant time, see \cite{ht-pav-08}.
  This concludes the linear-time implementation of our algorithm.
\end{proof}
 }{}

\section{Conclusion}\label{sec:conclusion}
In this paper, we gave a linear-time solution for the problem \peplan.
Our algorithm straightforwardly extends the well-known vertex-addition planarity tests of Booth and Lueker~\cite{boo-pta-75,bl-tft-76} and Haeupler and Tarjan~\cite{ht-pav-08}.
The core of our approach is the modification of the underlying data structure, the PC-tree, to also respect the constraints that stem from the partial drawing.
These constraints are derived from the insight that for an extension of the fixed partial drawing, it is necessary and sufficient to respect the fixed vertex rotations as well as to embed certain edges in predefined faces between fixed edges.
These constraints naturally translate from rotations of individual vertices to PC-trees that represent all planar embeddings of connected components.
Switching in these \pcs in the planarity test then requires us to handle two very specific situations more carefully than in the unconstrained setting.
We show that both can be easily resolved in our setting, thereby also uncovering assumptions the ordinary planarity test makes and providing new insights into the workings of the generic test and its~embedder.

In comparison to the previous decomposition-based approach by Angelini et al.\ \cite{abf-tpo-15}, this yields a strongly simplified algorithm that relies on one depth-first search together with a single non-trivial data structure.
The description of our core algorithm is less than half as long as the description of the algorithm by Angelini et al.\ \cite{abf-tpo-15}, while at the same time being far less technical\footnote{Note that the exposition of the core points of our algorithm is contained up to \Cref{sec:summary} on page \pageref{sec:summary}.
    In addition to the full proofs following said section, this work includes three pages of supplementary pseudo-code and an extensive background on PC-trees in \Cref{sec:pctree-details} to further improve accesiblity, while still being self-contained without these parts.}.
Moreover, our algorithm allows for a concise formulation in pseudo-code and thereby can be grasped in its entirety, which we see as our main contribution towards simplicity.
Note that the basic planarity test and the PC-tree as its underlying data structure already have practical implementations \cite{fpr-eco-21,bcpdb-smy-04,fin-cpa-24}.
Combining this with our detailed, implementation-level pseudo-code for the changes required in the partial case, yields a practical solution to \peplan.

An interesting question for future work is whether our algorithm can be extended to handle constraints to vertex rotations in the form of arbitrary \pcs that in particular were not derived from an instance of \peplan.
This would represent a natural generalization of the problem \textsc{(Partially) PQ-Constrained Planarity}\cite{br-spo-15,gkm-pta-08}, where ordinary PC-trees constrain the orders of (some of) the edges incident to a vertex.
The main complication here is that the constraints on an edge now no longer need to be the same at both its endpoints.

\appendix
\bibliography{bibliography}

\end{document}